\documentclass[12pt,preprint]{amsart}

\usepackage{amsmath, latexsym, amsfonts, amssymb, amsthm}
\usepackage{graphicx, color,hyperref,dsfont,epsfig,caption,wrapfig,subfig}
\usepackage{mathtools}
\usepackage{parskip}
\usepackage{movie15}
\hypersetup{
	linktoc=page,
	linkcolor=red,          
	citecolor=blue,        
	filecolor=blue,      
	urlcolor=cyan,
	colorlinks=true           
}

\NewDocumentEnvironment{eqs}{+b}
{\begin{equation}\begin{split}#1\end{split}\end{equation}}
{}
\numberwithin{equation}{section}

\newtheorem{theorem}{Theorem}[section]
\newtheorem{lemma}[theorem]{Lemma}
\newtheorem{proposition}[theorem]{Proposition}
\newtheorem{corollary}[theorem]{Corollary}

\newtheorem{remark}[theorem]{Remark}

\newcounter{conj}

\newtheorem{prob}[conj]{Problem}

\newcounter{thmc}

\theoremstyle{definition}

\renewcommand{\tilde}{\widetilde}          
\DeclareMathSymbol{\leqslant}{\mathalpha}{AMSa}{"36} 
\DeclareMathSymbol{\geqslant}{\mathalpha}{AMSa}{"3E} 
\DeclareMathSymbol{\eset}{\mathalpha}{AMSb}{"3F}     
\renewcommand{\leq}{\;\leqslant\;}                   
\renewcommand{\geq}{\;\geqslant\;}                   

\newcommand{\C}{\mathbb{C}}
\newcommand{\R}{\mathbb{R}}

\newcommand{\N}{\mathbb{N}}
 
\newcommand{\D}{\mathbb{D}} 
\newcommand{\Heps}{\mathbb{H}_{\delta,\eps}}
\newcommand{\Reps}{\mathbb{R}_{\eps}} 
\renewcommand{\H}{\mathbb{H}}
\newcommand{\im}{\bm{\mathrm{i}}}

\newcommand{\E}{\mathds{E}}
\newcommand{\X}{\bm{\mathrm X}}

\newcommand{\V}{\bm{\mathrm V}}

\newcommand{\ps}[1]{\langle #1 \rangle}
\newcommand{\mc}[1]{\mathcal{#1}}

\newcommand{\eqlaw}{\overset{\text{(law)}}{=}}

\makeatletter
\newcommand{\ostar}{\mathbin{\mathpalette\make@circled*}}
\newcommand{\make@circled}[2]{%
	\ooalign{$\m@th#1\smallbigcirc{#1}$\cr\hidewidth$\m@th#1#2$\hidewidth\cr}%
}
\newcommand{\smallbigcirc}[1]{%
	\vcenter{\hbox{\scalebox{0.77778}{$\m@th#1\bigcirc$}}}%
}
\makeatother

\def\Is{I_{sing} }
\def\Ir{I_{reg} }
\def\It{I_{tot} }

\def\Itr{I_{tot}^{\delta,\eps} }

\newcommand{\qt}[1]{\quad\text{#1}\quad}

\def\V{\bm{\mathrm V}}

\def\bmD{\bm{\mathrm D}}

\def\Lc{\bm{\mathcal L}}
\def\X{\bm{\mathrm  X}}

\def\eps{\varepsilon}

\renewcommand{\k}{\mathbf{k}}

\renewcommand{\d}{\text{\rm d}}
\newcommand{\Herm}{\bm{\mathrm{He}}}

\def\Div{\bmD_{\bm z,\bm a}}
\def\Diva{\bmD_{\bm z,\bm\alpha}}

\def\eul{\chi(\Sigma,\Div)}
\def\Hdiv{H_{\bm z,\bm a}}
\def\Hdiva{H_{\bm z,\bm \alpha}}

\def\gdiv{g_{\bm z,\bm a}}
\def\phidiv{\phi_{\bm z,\bm a}}
\def\Phidiv{\Phi_{\bm z,\bm a}}
\def\vphidiv{\varphi_{\bm z,\bm a}}
\def\phireg{\phi_{\delta,\eps}}
\def\Phireg{\Phi_{\delta,\eps}}
\def\vphireg{\varphi_{\delta,\eps}}
\def\tdiv{\tau(\Sigma,\bm a)}

\def\eps{\varepsilon}

\def\bi{\begin{itemize}}
	\def\ei{\end{itemize}}

\def\bnum{\begin{enumerate}}
	\def\enum{\end{enumerate}}

\def\<#1{\langle #1 \rangle}
\def\eval#1#2#3#4{\left[#4\right]^{#1}_{#2}\hspace*{-0.6cm}\vcenter{\hbox{$\scriptstyle #3$}}\hspace*{0.2cm}}

\def\red#1{\textcolor{red}{#1}}

\newcommand{\norm}[1]{\left\lvert#1\right\rvert}

\newcommand{\nnorm}[1]{\left\lvert\left\lvert#1\right\rvert\right\rvert}
\newcommand{\expect}[1]{\mathbb{E}\left[#1\right]}
\usepackage{bm}
\usepackage{bbm}

\title[Boundary Liouville theory: classical and semi-classical]{Around the semi-classical limit of boundary Liouville conformal field theory}

\author{Baptiste Cercl\'e}
\email{baptiste.cercle@epfl.ch}
\address{EPFL SB MATH RGM, MA B2 397, Station 8, CH-1015 Lausanne, Switzerland.}

\begin{document}
	
	\maketitle
	\begin{abstract}
		Liouville conformal field theory describes a random geometry that fluctuates around a deterministic one: the unique solution of the problem of finding, within a given conformal class, a Riemannian metric with prescribed scalar and geodesic curvatures as well as conical singularities and corners. The level of randomness in Liouville theory is measured by the coupling constant $\gamma\in(0,2)$, the semi-classical limit corresponding to taking $\gamma\to0$.
		
		Based on the probabilistic definition of Liouville theory, we prove that this semi-classical limit exists and does give rise to this deterministic geometry. At second order this limit is described in terms of a massive Gaussian free field with Robin boundary conditions. 
		
		This in turn allows to implement CFT-inspired techniques in a deterministic setting: in particular we define the classical stress-energy tensor, show that it can be expressed in terms of accessory parameters (written as regularized derivatives of the Liouville action), and that it gives rise to classical higher equations of motion. 
	\end{abstract}
	
	

	\section{Introduction}
	
	\subsection{Boundary Liouville theory: classical and quantum}
	
	Liouville theory is a two-dimensional Conformal Field Theory (CFT hereafter) introduced by Polyakov in 1981~\cite{Pol81}. In this fundamental article a notion of random geometry is defined based on the consideration of the \textit{Liouville equation} $-\Delta u=e^u$, which naturally appears in numerous classical problems such as uniformization of Riemann surfaces or the Nirenberg problem. 
	
	\subsubsection{Prescription of curvatures and conical singularities}
	Indeed the classical question inherent to Liouville CFT is of geometric nature: it is the problem of finding, within a given conformal class, a Riemannian metric with prescribed curvatures and conical singularities. Namely let $(\Sigma,g)$ be a compact, connected Riemannian surface equipped with a smooth Riemannian metric $g$ and denote its boundary by $\partial\Sigma$. We also fix a divisor 
	\[
	\Div=\sum_{k=1}^Na_k z_k+\sum_{l=1}^M b_l s_l\qt{where $a_k,b_l>-1$ and with $z_k\in\Sigma$ and $s_l\in\partial\Sigma$ all distinct.}
	\]
	The problem under consideration is then (see Section~\ref{sec:unif} for a more precise statement): 
	\begin{prob}\label{prob:class_intro}
		Find a conformal Riemannian metric $\gdiv=e^{\phi}g$ on $\Sigma$ such that:
		\begin{itemize}
			\item $\gdiv$ is smooth away from $\bm z\coloneqq\text{supp}(\Div)$, has a conical singularity of angle $2\pi(1+\alpha_k)$ at $z_k$ for any $1\leq k\leq N$ and a corner of angle $\pi(1+b_l)$ at $s_l$, $1\leq l\leq M$; 
			\item $\gdiv$ has Gaussian curvature $-\frac12\Lambda$ inside $\Sigma\setminus\bm z$, with $\Lambda \geq 0$ ($\Lambda\not\equiv0$) smooth on $\Sigma$;
			\item $\gdiv$ has geodesic curvature $-\sigma$ on $\partial\Sigma\setminus\bm z$, with $\sigma\geq 0$ smooth and bounded on $\partial\Sigma\setminus\bm z$.
		\end{itemize}
	\end{prob}
	In particular $\sigma$ can have discontinuity points at the $s_l$'s. Existence of solutions to this question is constrained by the Gauss-Bonnet formula (see Equation~\eqref{eq:gauss_bonnet_sing}) which implies that the singular Euler characteristic $\eul=\chi(\Sigma)+\sum_{k=1}^N\alpha_k+\frac12\sum_{l=1}^Mb_l$ must be negative. It is actually a sufficient condition and under this assumption Problem~\ref{prob:class_intro} admits a unique solution as stated in Theorem~\ref{thm:prescribe} below.
	
	The problem of prescribing curvature(s) on a Riemann surface is a very classical one: in its simplest form it asks for a solution of the  Liouville equation $\Delta u=e^u$, a question asked by the G\"ottingen Royal Society of Sciences and answered by Picard in 1890~\cite{Picard1} (see also~\cite{Picard2, Picard3}). As demonstrated by Poincaré in 1898~\cite{Poincaré}, this question has a particular relevance in the context of uniformization of Riemann surfaces as it allows, when the singularities are \textit{parabolic} (corresponding to the critical case $a_k=-1$, also called \textit{cusps} or in the words of Poincaré \lq\lq sommets de la troisième espèce"), to construct an analytic covering of a punctured sphere by the upper-half plane. The curvature in these cases is then taken to be constant: the Nirenberg problem is more generally concerned with the question of finding conditions under which a function on a Riemann surface $\Sigma$ is the curvature of some Riemannian metric on $\Sigma$. This question was addressed in~\cite{Berger, KW_curv} in the case of (closed) compact surfaces and in~\cite{LSMR} in the presence of a boundary, in which case one prescribes both scalar and geodesic curvatures. When in addition the underlying surface is allowed to have conical singularities, in the case where $\Sigma$ is either closed or with geodesic boundary ($\sigma\equiv0$) this problem was solved by Troyanov~\cite{Troyanov} while the general case was settled recently in~\cite{BRS}.
	
	Under the simplifying assumptions made in the present paper, that is that the curvatures are non-positive and that $\eul$ is negative, a classical approach to address Problem~\ref{prob:class_intro} is to use a variational formulation based on the Liouville action (already considered in Poincaré's essay~\cite[Chapitre IX]{Poincaré}), which is formally defined (see Proposition~\ref{prop:def_action} for a proper statement) by setting for $\phi$ smooth away from supp$(\Div)$ and for $g$ a given smooth Riemannian metric on $\Sigma$ (with scalar and geodesic curvatures $R_g$ and $k_g$, volume form and line element $\d v_g$ and $\d l_g$, gradient and norm on the tangent space $\nabla_g$ and $\norm{\cdot}_g$): 
	\begin{eqs}\label{eq:action}
		S_{\bm z,\bm a}(\phi)\quad\lq\lq\coloneqq"\quad &\frac1{4\pi}\int_\Sigma\left(\norm{\nabla_g \phi}_g^2+2R_g\phi + 2\Lambda e^{\phi}\right)\d v_g+\frac1{2\pi}\int_{\partial\Sigma}\left(2k_g\phi + 4\sigma e^{\phi/2}\right)\d l_g\\
		&+\sum_{k=1}^N2a_k\phi(z_k)+\sum_{l=1}^Mb_l\phi(s_l).
	\end{eqs}
	Indeed if $\phidiv$ is a critical point (actually a minimum) of this action then $e^{\phidiv}g$ solves Problem~\ref{prob:class_intro}. In particular using the Liouville action as a Gibbs potential naturally provides a (formal) way of defining a notion of random geometry on $\Sigma$ that will resemble the classical one: this is the path integral definition of Liouville Conformal Field Theory (CFT hereafter).

	\subsubsection{Boundary Liouville CFT}
	To be more specific, let $\gamma>0$ (called \textit{coupling constant}) and formally define a Borel measure on the Sobolev space with negative index $\mc F\coloneqq H^{-1}(\Sigma,g)$ by setting for any continuous and bounded functional $F:\mc F\to\R$
	\begin{equation}\label{eq:path}
		\ps{F}_{\gamma,\bm z,\bm a}\quad\lq\lq\coloneqq" \quad\frac1{\mc Z_g}\int_{\mc F}F(\phi)e^{-S_{\gamma,\bm z,\bm a}(\phi)}\d\phi.
	\end{equation}
	Here $\mc Z_g$ is a normalization constant given by a (regularized) determinant, $\d\phi$ plays the role of a uniform measure over $\mc F$ and $S_{\gamma,\bm z,\bm a}(\phi)$ is (almost) given by $\frac{1}{\gamma^2}S_{\bm z,\bm a}(\gamma\phi)$:
	\begin{eqs}\label{eq:action_CFT}
		S_{\gamma,\bm z,\bm a}(\phi)\coloneqq &\frac1{4\pi}\int_\Sigma\left(\norm{\nabla_g \phi}_g^2+QR_g\phi + 2\mu e^{\gamma\phi}\right)\d v_g+\frac1{2\pi}\int_{\partial\Sigma}\left(Qk_g\phi + 4\mu_\partial e^{\gamma\phi/2}\right)\d l_g\\
		&-\sum_{k=1}^N\alpha_k\phi(z_k)-\sum_{l=1}^M\frac12\beta_l\phi(s_l)
	\end{eqs}
	where $Q=\frac2\gamma+\frac\gamma2$ is the \textit{background charge} and where we have set $\mu\coloneqq \frac{\Lambda}{\gamma^2}$, $\mu_\partial\coloneqq \frac{\sigma}{\gamma^2}$ and $\alpha_k\coloneqq -\frac{2a_k}{\gamma}$, $\beta_l\coloneqq -\frac{2b_l}{\gamma}$. We will denote $\ps{F(\Phi)\prod_{k=1}^NV_{\alpha_k}(z_k)\prod_{l=1}^MV_{\beta_l}(s_l)}_{\gamma,\mu,\mu_\partial}\coloneqq \ps{F}_{\gamma,\bm z,\bm a}$: for $F=1$ these are called the \textit{correlation functions of Vertex Operators}.
	
	The previous writing being purely formal, a proper mathematical definition of the path integral~\eqref{eq:path} is necessary. This is the starting point of the program initiated by David-Kupiainen-Rhodes-Vargas~\cite{DKRV} who defined Liouville correlation functions based on a probabilistic framework that involves Gaussian Free Fields and its exponential: Gaussian Multiplicative Chaos measures. First carried out on the sphere, this construction was then developed to take into account all possible two-dimensional geometries: higher genus in~\cite{DRV16, GRV16} and open surfaces in~\cite{HRV16, Wu}. This probabilistic take on Liouville CFT has proved to be extremely successful in many perspectives and in particular in the context of two-dimensional CFT~\cite{BPZ} via the rigorous derivation of predictions made in the physics literature, starting from Ward identities~\cite{KRV_loc, Cer_HEM}, computation of the structure constants~\cite{KRV_DOZZ, ARS, ARSZ} recovering predictions made in~\cite{DO94, ZZ96, FZZ, Hos, PT02}, implementation of the conformal bootstrap procedure~\cite{GKRV} and of Segal's axioms~\cite{Seg04} in~\cite{GKRV_Segal, GRW}, study of the structure of the Virasoro modules~\cite{BGKRV, BW_irr, BaWu} and of the conformal blocks thus defined~\cite{BGKR}...
	
	\subsubsection{From Liouville CFT to the classical theory}
	Though the path integral definition for Liouville CFT is purely formal, it still gives a lot of insight on the actual properties of the mathematical model thus defined. For instance a handwavy application of Laplace's method in Equation~\eqref{eq:path} shows that one should expect that as $\gamma\to0$ the integral concentrates one the minimum $\phidiv$ of the Liouville action and that we have the following asymptotic:
	\begin{equation}\label{eq:semi_heur}
		\ps{F}_{\gamma,\bm z,\bm a}\sim e^{-\frac1{\gamma^2}S_{\bm z,\bm a}(\phidiv)}\frac1{\mc Z_g}\int_{\Sigma} F\left(\phi+\frac1\gamma\phidiv\right)e^{-\frac12 S''_{\bm z,\bm a}(\phi)}\d\phi
	\end{equation}
	where $S''_{\bm z,\bm a}(\phi)\coloneqq\left.\partial_t^2\right\vert_{t=0}S_{\bm z,\bm a}(\phidiv+t\phi)$, which is found to be given by $\ps{\phi, D_{\Lambda,\sigma}\phi}_{\bm z,\bm a}$ with
	\begin{equation}
		\ps{u, D_{\Lambda,\sigma}v}_{\bm z,\bm a}\coloneqq\frac1{2\pi}\int_\Sigma u\left(\Lambda e^{\phidiv}-\Delta_g\right)v\d v_g+\frac1{2\pi}\int_{\partial\Sigma}u\left(\sigma e^{\frac12\phidiv}+\partial_{n_g}\right)v\d l_g.
	\end{equation}
	
	In particular this heuristic allows to import CFT techniques in the classical theory: an important application of the above is Polyakov's prediction for a question raised by Poincaré~\cite{Poincaré2, Poincaré}. Namely if $\psi:\H\to X$ is an analytic covering of the $N$-punctured sphere $X=\C\setminus\{z_1,\cdots,z_{N-1}\}$ then the Schwarzian derivative of its (multi-valued) inverse satisfies
	\begin{equation}
		\left\{\psi^{-1},z\right\}=\sum_{k=1}^{N-1}\left(\frac{1}{2(z-z_k)^2}+\frac{c_k}{z-z_k}\right)\qt{and we have}\left(\partial^2_z+\left\{\psi^{-1},z\right\}\right)e^{-\frac\phidiv2}=0
	\end{equation}
	where the coefficients $c_k$ are called \textit{accessory parameters} (and closely related to Poincaré's \lq\lq invariants fondamentaux"~\cite[Chapitre I]{Poincaré2}), while $\phidiv$ is the solution of Liouville equation on $X$ with parabolic singularities at the punctures. The explicit construction of a covering map $\H\to X$ is therefore intimately related to the identification of these accessory parameters, which can thus be achieved via the definition of the \textit{classical stress-energy tensor}
	\begin{equation}
		T(z)\coloneqq \partial^2_z\phidiv(z)-\frac12\left(\partial_z\phidiv(z)\right)^2=\frac12\left\{\psi^{-1},z\right\}.
	\end{equation}
	Based on the $\gamma\to0$ limit of the Ward identities for Liouville correlation functions (themselves expressed in terms of the \textit{quantum stress-energy tensor}), Polyakov was then able to predict a simple expression for the accessory parameters. Namely that $c_k=-\frac12\partial_{z_k}S_{\bm z,\bm a}(\phidiv)$, a statement that was proved by Zograf-Takhtajan~\cite{TaZo1} (using completely different methods).

	A rigorous implementation of Polyakov's heuristic, based on the probabilistic definition of Liouville CFT, has been carried out in~\cite{LRV_semi1, LRV_semi2} where the proposed relation $c_k=-\frac12\partial_{z_k}S_{\bm z,\bm a}(\phidiv)$ was obtained as a corollary of the semi-classical limit $\gamma\to0$ of Liouville correlation functions. More precisely Lacoin-Rhodes-Vargas provide in~\cite{LRV_semi2} a rigorous derivation of the asymptotic~\eqref{eq:semi_heur} when the underlying surface is the Riemann sphere, with the prescribed curvature chosen constant and where the singularities satisfy $a_k<-1$ and $\eul<0$ (in which case accessory parameters were computed in~\cite{TaZo2}). 
	On a similar perspective the semi-classical limit of the conformal blocks on the torus has been recently rigorously derived in~\cite{DGP} using this probabilistic framework, showing at leading order a behavior similar to Equation~\eqref{eq:semi_heur} (though the multiplicative constant was left unidentified there).
	
	\subsection{Main results}
	We are primarily interested in the present paper in extending the results of~\cite{LRV_semi2} to any compact, connected, smooth Riemannian surface (especially with non-empty boundary), and with non-constant prescribed curvature. As a counterpart this allows to implement CFT techniques in the classical setting of Problem~\ref{prob:class_intro} and by doing so to uncover previously unknown (to the best of our knowledge) phenomenons, deterministic counterparts of the \textit{higher equations of motion} encountered in the CFT literature~\cite{BB10, BaWu, Cer_HEM}. This sheds light on some of the singular aspects of Problem~\ref{prob:class_intro} in the presence of a boundary. 
	
	\subsubsection{Semi-classical limit of boundary Liouville theory}
	Our first result is the rigorous derivation of the semi-classical limit $\gamma\to0$ of Liouville CFT, heuristically discussed above:
	\begin{theorem}\label{thm:semi_classical_intro}
		Assume that $\eul<0$ and let $\gdiv=e^{\phidiv} g$ be the unique solution of Problem~\ref{prob:class_intro}. Then for any $F$ continuous bounded over $H^{-1}(\Sigma,g)$, as $\gamma\to0$:
		\begin{eqs}\label{eq:semi_classical_intro}
			&\ps{F\left(\Phi-\frac1\gamma\phidiv\right)\prod_{k=1}^{N}V_{-\frac{2a_k}\gamma}(z_k)\prod_{l=1}^{M}V_{-\frac{2\beta_l}\gamma}(s_l)}_{\gamma,\frac{\Lambda}{\gamma^2},\frac{\sigma}{\gamma^2}}\sim e^{-\frac1{\gamma^2}S_{\bm z,\bm a}} [F]_0\qt{with}\\
			&[F]_0=[1]_0\int_\R\E_{\bm z,\bm a}\left[F(\X_{\bm z,\bm a}+c)\right]e^{-c^2\ps{1}_{\bm z,\bm a}}\sqrt{\frac{\ps{1}_{\bm z,\bm a}}{\pi}}\d c
		\end{eqs}
		where $\X_{\bm z,\bm a}$ is a massive Gaussian Free Field with mass $\Lambda$ and Robin boundary conditions $\left(\sigma e^{\frac12\phidiv}+\partial_{n_{g}}\right)\X_{\bm z,\bm a}=0$ (\textit{i.e.} has covariance kernel $G_{\bm z,\bm a}$ defined in Lemma~\ref{lemma:massive_GMC}).
	\end{theorem}
	In the above $\Lambda$ and $\sigma$ satisfy the assumptions of Problem~\ref{prob:class_intro} (and in particular are not necessarily constant). This proves that the heuristic from Equation~\eqref{eq:semi_heur} is indeed valid, and that we can write an expansion (in distribution) as $\gamma\to0$ of the Liouville field of the form $\Phi=\frac1\gamma\phidiv+\X_{\bm z,\bm a}+C$ where $C$ is an independent Gaussian random variable.

    In a recent preprint~\cite{LMSWY2}, the semi-classical limit of the so-called \textit{$6j$-symbols} for the principal series of the modular double of $U_q\mathfrak{sl}(2; \R)$ has been obtained. These $6j$-symbols appear in various areas of mathematical physics and notably in the study of the Virasoro TQFT and three-dimensional quantum Anti-de Sitter gravity, see~\cite{LMSWY1, LMSWY2} and the references therein for more context. It was brought to our attention by the authors of~\cite{LMSWY2} that combining the statement of Theorem~\ref{thm:semi_classical_intro} together with~\cite[Theorem 1.1]{ARSZ} and~\cite[Theorems 1.3 and 1.4]{LMSWY2} thus gives a relation between the semi-classical limit of such $6j$-symbols (that arises in the setting of 3d hyperbolic and AdS geometry) and the semi-classical limit of Liouville theory (that is the classical Liouville action in view of Theorem~\ref{thm:semi_classical_intro}), reminiscent of the AdS/CFT correspondence. We refer to~\cite[Section 1.5]{LMSWY2} for more details on this connection.

	\subsubsection{Classical implications}
	Thanks to this statement we can import some results and techniques from the probabilistic study of Liouville CFT. For the sake of simplicity we assume that $\Sigma$ is the upper half-plane $\H$ and consider $\vphidiv$ the unique solution of \begin{equation}\label{eq:PDE_intro}
		\left\lbrace \begin{array}{ll}
			-\Delta \Phidiv= \Lambda e^{\Phidiv} & \text{in } \H \\
			\partial_{n}\Phidiv= -2\sigma e^{\frac12\Phidiv} & \text{on }\R
		\end{array} \right.
	\end{equation}
	with the behavior near a singular point $\Phidiv(x)\sim -2a_k\ln\norm{x-z_k}$ (and likewise for boundary punctures) and where $\Lambda$ (resp. $\sigma$) are now chosen to be constant negative (resp. constant and non-positive on each connected component of $\partial\Sigma\setminus\bm z$).
	Then unlike in the closed case, defining the classical stress-energy tensor by setting $T(z)\coloneqq \partial_z^2\Phidiv(z)-\frac12\left(\partial_z\Phidiv(z)\right)^2$ requires some extra care when $z$ lies on the boundary of $\Sigma$ since the above expression is actually ill-defined. To give it a proper meaning one needs to go through a limiting procedure involving a regularization $\Phireg$ of $\Phidiv$ obtained by \lq\lq smoothing" the conical singularities. This is achieved by changing $\Lambda$ to $\Lambda_{\delta,\eps}$ where the latter is smooth, zero near $\bm z\cup\partial\Sigma$ and coincides with $\Lambda$ away from $\bm z\cup\partial\Sigma$ (and likewise for $\sigma$, see Subsection~\ref{subsec:phireg} for more details):
	\begin{theorem}\label{thm:SET_intro}
		For a bulk point $z\in\H\setminus\bm z$ set $T(z)\coloneqq \partial_z^2\Phidiv(z)-\frac12\left(\partial_z\Phidiv(z)\right)^2$, while on the boundary define for any $t\in\R\setminus\bm z$ (where by Lemma~\ref{lemma:L2} the limit exists):
		\begin{equation}
			T(t)\coloneqq-\frac\Lambda4 e^{\Phidiv(t)}+\lim\limits_{\delta,\eps\to0}\left(\frac12\partial_t^2\Phireg(t)-\frac18\left(\partial_t\Phireg(t)\right)^2+\frac{1}{\eps}\frac{\sigma(t)}{2\pi}e^{\frac12\Phireg(t)}\right).
		\end{equation}
		Moreover define weights by putting $\delta_{a}\coloneqq-a(1+\frac12{a})$ for $a>-1$. Then for $x\in\overline \H\setminus\bm z$
		\begin{equation}
			T(x)=\sum_{k=1}^N2\mathfrak{Re}\left(\frac{\delta_{a_k}}{(x-z_k)^2}+\frac{\bm c_k}{x-z_k}\right)+\sum_{l=1}^M\left(\frac{\delta_{b_l}}{(x-s_l)^2}+\frac{\bm c_l}{x-s_l}\right)
		\end{equation}
		where the accessory parameters are given by (weak) derivatives of the Liouville action with respect to the location of the punctures: for a bulk puncture $\bm c_k=-\frac12\partial_{z_k}S_{\bm z,\bm a}$ while 
		\begin{equation}\label{eq:access_intro}
			c_l=-\frac12\partial_{s_l}S_{\bm z,\bm a} \coloneqq-\frac12\lim\limits_{\delta,\eps\to0}\left(\partial_{s_l}S_{\delta,\eps}-\eval{s_l+\eps}{s_l-\eps}{-}{2\sigma(t) e^{\frac12\phireg(t)}}\right)
		\end{equation}
		for a boundary one (Proposition~\ref{prop:accessory}). Finally in the weak sense of derivatives we have:
		\begin{equation}\label{eq:HEM_class_intro}
			\left\lbrace \begin{array}{ll}
				\left(\partial_z^2+\frac12 T(z)\right)e^{-\frac12\Phidiv(z)}=0&\text{in }\H\setminus\bm z\\
				\left(\partial_t^2+\frac12 T(t)\right)e^{-\frac14\Phidiv(t)}=\frac14\left(\sigma(t)^2-\frac{\Lambda}{2}\right)e^{\frac34\Phidiv(t)}&\text{on }\R\setminus\bm z.
			\end{array} \right.
		\end{equation}
	\end{theorem}
	We actually provide in Theorem~\ref{thm:L2} a more general definition of the stress-energy tensor at a singular point based on a (slightly more involved) regularization procedure, thus allowing for a direct translation of CFT techniques within our setting. For instance if we consider only boundary insertions, with $s_l$ having the special weight $b_l=1$ we get the equation:
	\begin{equation}\label{eq:HEM_heavy_intro}
		\frac12\left(\partial_{s_l}S_{\bm z,\bm a}\right)^2+\sum_{k\neq l}\frac{4\delta_k}{(s_l-s_k)^2}-\frac{2}{s_l-s_k}\partial_{s_k}S_{\bm z,\bm a}=-2\frac{\sigma(s_l^-)+\sigma(s_l^+)}{\pi} e^{\frac12\Phidiv^{(l)}(s_l)}.
	\end{equation}
	This equation is reminiscent of the differential equation obtained in~\cite[Proposition 1.8]{PeWa}, for the minimal (multichordal) Loewner potential, when $\sigma(s_l^-)+\sigma(s_l^+)=0$ for all $l$. 
	
	The last equation in~\eqref{eq:HEM_class_intro} is referred to as an \textit{higher equation of motion} in the CFT literature~\cite{BB10, BaWu, Cer_HEM}, and can be heuristically justified as follows: for $z=t+i\delta$ with $\delta\to0$ and since $\Phidiv$ satisfies~\eqref{eq:PDE_intro}, if $\Phidiv$ was differentiable on $\R\setminus\bm z$ we would have up to a $o(\delta)$
	\[
	\left(\partial_{z}^2+\frac12 T(z)\right) e^{-\frac12\Phidiv(z)}=\left[\left(\partial_t^2+\frac12 T(z)\right)e^{-\frac14\Phidiv(z)}-\frac14\left(\sigma(t)^2-\frac\Lambda2\right)e^{\frac34\Phidiv}\right]e^{-\frac14\Phidiv(z)}.
	\]
	However since this is not the case as soon as the boundary is not geodesic the latter needs a proper justification. The same applies to the accessory parameters, since in order to define them on for a boundary insertion one relies on a limiting procedure too, highlighting the qualitative difference of behavior of $\Phidiv$ on the boundary compared with inside the bulk.

	After a first draft of the manuscript was written, it was brought to our attention that the expression of the accessory parameters given in Equation~\eqref{eq:access_intro} was already derived in~\cite{HJ06}, where a similar regularization procedure was used.
    However to the best of our knowledge, the existence of higher equations of motions in the form of~\eqref{eq:HEM_class_intro} or~\eqref{eq:HEM_heavy_intro} for the classical field, or the limiting procedure needed to define the stress-energy tensor, are novel as far as we know. We hope to understand in more depth the geometrical meaning of such equations and explore some of its implications in the future.
    
	\textbf{\textit{Acknowledgments}}
	The author would like to thank Guillaume Baverez, Colin Guillarmou, Eveliina Peltola and Baojun Wu for discussions related (or not) to this work. The author is also grateful to Xin Sun and Peking University for their hospitality. \\
	Part of the manuscript has been written during the trimester program \lq\lq Probabilistic methods in quantum field theory" organized at the Hausdorff Research Institute for Mathematics and funded by the Deutsche Forschungsgemeinschaft (DFG, German Research Foundation) under Germany's Excellence Strategy – EXC-2047/1 – 390685813.\\
	The author has been supported by Eccellenza grant 194648 of the Swiss National Science Foundation and is a member of NCCR SwissMAP.

	
	
	\section{Prescription of curvatures and conical singularities}\label{sec:unif}
	We describe in this first section the classical problem under consideration in this document and which can be formulated as follows: given $\Sigma$ an open Riemann surface equipped with a smooth Riemannian metric $g$, can we find a conformal metric with constant scalar curvature, piecewise constant geodesic curvature and prescribed conical singularities and corners?
	
	
	
	\subsection{Riemannian geometry and conical singularities}
	Let $\Sigma$ be a compact connected Riemann surface with non-empty boundary $\partial\Sigma$, and let $\chi(\Sigma)$ be its Euler characteristic. For notational simplicity we make the convention that $\partial\Sigma\cap\Sigma=\emptyset$ and set $\overline\Sigma=\Sigma\sqcup\partial\Sigma$. We assume that $\overline\Sigma$ is equipped with a smooth Riemannian metric $g$, with corresponding Ricci scalar curvature $R_g$ and geodesic curvature $k_g$. We then have the Gauss-Bonnet formula:
	\begin{equation} \label{eq:gauss-bonnet}
		\frac1{4\pi}\int_\Sigma R_g \d v_g +\frac1{2\pi}\int_{\partial\Sigma} k_g \d \lambda_g = \chi(\Sigma).
	\end{equation}
	We also denote by $\Delta_g$ the Laplacian in the metric $g$ and $\partial_{\vec{n}_g}$ the outward normal derivative. 
	
	\subsubsection{Uniformization of open Riemann surfaces}
	Given $g$ a Riemannian metric on $\Sigma$ its conformal class is the set of metrics $g'=e^{\varphi}g$ with $\varphi$ smooth over $\overline{\Sigma}$. Under such a change of metric the variations of the scalar and geodesic curvatures are described by
	\begin{equation} \label{eq:conf_curv}
		R_{g'} = e^{-\varphi}(R_g - \Delta_g\varphi) \quad ; \quad k_{g'} = e^{-\varphi/2}(k_g + \frac12\partial_{\vec{n}_g} \varphi).
	\end{equation}
	Within each conformal class there are~\cite{OPS88} so-called uniform metrics of type $I$ and $II$, unique up to scaling and isometry. They are such that in these metrics the scalar curvature is constant (resp. $0$) while the geodesic curvature is $0$ (resp. constant) for type $I$ (resp. $II$). The sign of these constants is prescribed by the Euler characteristic of $\Sigma$.
	
	\subsubsection{Conical singularities and corners}
	Let $z_1,\cdots,z_N$ (resp.  $s_1,\cdots,s_M$) be distinct points in $\Sigma$ (resp. on $\partial\Sigma$), and set $\bm z\coloneqq\{z_1,\cdots,z_N,s_1,\cdots,s_M\}$. To these punctures we associate real numbers $a_1,\cdots,a_N$ and $b_1,\cdots,b_M$ subject to the condition
	\begin{equation}\label{eq:sei_cla}
		a_k>-1\qt{and}b_l>-1\qt{for all}1\leq k\leq N\qt{and}1\leq l\leq M.	
	\end{equation}
	Let us write $\bm a\coloneqq\{a_1,\cdots,a_N,b_1,\cdots,b_M\}$ and form the divisor
	\begin{equation}\label{eq:divisor}
		\Div\coloneqq \sum_{k=1}^Na_kz_k+\frac12\sum_{l=1}^Mb_ls_l.
	\end{equation}
    Due to these singularities we will need to cut out small (semi-) disks from $\Sigma$. To do so, for $\delta,\eps>0$ small enough and $g$ smooth let $\Sigma_\delta\coloneqq\{x\in\Sigma,d_g(x,\partial\Sigma)>\delta\}$ and $\bm z_\eps\coloneqq \cup_{z\in\bm z}B_g(z,\eps)$ where $B_g(z,\eps)$ is the geodesic (semi-) disk of radius $\eps$ centered at $z$. We then set
	\begin{equation}
		\Sigma_{\delta,\eps}\coloneqq\Sigma_\delta\setminus\bm z_\eps\qt{and}\left(\partial\Sigma\right)_{\eps}\coloneqq\partial\Sigma\setminus\bm z_\eps.
	\end{equation}
	
	In agreement with~\cite{Troyanov}, we will say that \textit{$g_s$ represents the divisor $\Div$} if $g_s$ is a Riemannian metric on $\overline\Sigma\setminus\bm z$ such that for each singular point $z_k\in\bm z$, there is an open neighbourhood $\mc O\subset \overline{\Sigma}$ of $z_k$, local coordinates $x$ with $x(z_k)=0$ and $w$ continuous over $\mc O$ for which 
	\begin{equation}\label{eq:gadm}
		g_s(x)=\norm{x}^{2 a_k}e^{2w(x)}\norm{dx}^2.
	\end{equation}
	When $z_k$ is in $\Sigma$ this amounts to saying that $g_s$ has a conical singularity of order $a_k$ (or angle $2\pi(1+a_k)$) at $z_k$. If $z_k$ is on $\partial\Sigma$ then $g_s$ has a corner of order $a_k$ (or angle $\pi(1+a_k)$). 
	Hereafter we will say that such a $g_s$ is \textit{admissible} when it satisfies the following assumptions:
	\begin{itemize}
		\item $g_s$ represents the divisor $\Div$;
		\item $R_{g_s}$ extends to a continuous function over $\Sigma$;
		\item $k_{g_s}$ is continuous on each connected component of $\partial\Sigma\setminus\bm z$.
	\end{itemize}	
	Under this assumption, let us define the (singular) Euler characteristic of $\left(\Sigma,\bmD_{\bm z,\bm a}\right)$ by setting
	\begin{equation}\label{eq:eul}
		\chi\left(\Sigma,\bmD_{\bm z,\bm a}\right)\coloneqq \chi(\Sigma)+\sum_{k=1}^Na_k+\frac12\sum_{l=1}^Mb_l.
	\end{equation}
	Then we have the following generalized Gauss-Bonnet formula:
	\begin{proposition}\label{prop:gauss-bonnet}
		Assume that $g_s$ is admissible and represents the divisor $\bmD_{\bm z,\bm a}$. Then
		\begin{equation}\label{eq:gauss_bonnet_sing}
			\frac1{4\pi}\int_\Sigma R_{g_s}\d v_{g_s}+\frac1{2\pi}\int_{\partial\Sigma}k_{g_s} \d l_{g_s}=\chi\left(\Sigma,\Div\right).
		\end{equation}
	\end{proposition}
	\begin{proof}
		The argument is similar to~\cite[Proposition 1]{Troyanov}. We write $g_s=e^{\varphi}g_0$ where $g_0$ is a uniform type $I$ metric on $\Sigma$. Then thanks to Equation~\eqref{eq:conf_curv} together with Equation~\eqref{eq:gauss-bonnet}
		\begin{align*}
			\chi\left(\Sigma,\Div\right)&=\chi\left(\Sigma\right)+\frac1{4\pi}\lim\limits_{\eps\to0}\int_{\Sigma_{0,\eps}} -\Delta_{g_0}\varphi\d v_{g_0}+\int_{\left(\partial\Sigma\right)_\eps}\partial_{n_{g_0}}\varphi \d l_{g_0}\\
			&=\chi\left(\Sigma\right)+\frac1{4\pi}\lim\limits_{\eps\to0}\int_{\partial\left(\Sigma_{0,\eps}\right)} -\partial_{n_{g_0}}\varphi\d l_{g_0}+\int_{\left(\partial\Sigma\right)_\eps}\partial_{n_{g_0}}\varphi \d l_{g_0}.
		\end{align*}
		Now for $\eps$ small enough $\partial\left(\Sigma_{0,\eps}\right)$ is the disjoint union of $\left(\partial\Sigma\right)_\eps$ and $\cup_{z\in\bm z}\partial B_{g_0}(z,\eps)$. Hence
		\begin{align*}
			\chi\left(\Sigma,\Div\right)&=\chi\left(\Sigma\right)-\frac1{4\pi}\sum_{z\in\bm z}\lim\limits_{\eps\to0}\int_{\partial B_{g_0}(z,\eps)}\partial_{n_{g_0}}\varphi\d l_{g_0}
		\end{align*}
		with $\vec n_{g_0}$ pointing inside $B_{g_0}(z,\eps)$. In local coordinates $\varphi(x)=2a\ln\norm{x}+w(x)$ where for a bulk insertion $\eps\partial_nw(\xi)\to0$ uniformly over $\xi\in\partial B_{g_0}(z,\eps)$ via~\cite[Lemma 3]{Troyanov}. In the boundary case, in a neighborhood of $z$, $\partial_{n_{g_0}} w$ is integrable since it is locally bounded (in local coordinates) by $C\norm{x}^{a}$, thus we also have the uniform bound $\eps\partial_nw(\xi)\to0$ over $\partial B_{g_0}(z,\eps)$.
	\end{proof}
	Hereafter to lighten the notations we will often denote $a_{N+l}\coloneqq \frac12b_l$ and $z_{N+l}\coloneqq s_l$. For instance the singular Euler characteristic becomes $\eul=\chi(\Sigma)+\sum_{k=1}^{N+M}a_k$.

	\subsubsection{Green's functions} \label{subsec:green}
	There is a natural way to express a singular metric as above using Green's functions. In the closed case they are defined as solutions of the  (weak) problem
	\begin{equation}\label{eq:closed_pb_green}
		\left\lbrace \begin{array}{ll}
			-\Delta_g G_g(\cdot,y) = 2\pi\left(\delta_y - \frac{1}{v_g(\Sigma)}\right) & \text{in } \Sigma \\
			\int_\Sigma G_g(x,y)\d v_g(x) = 0 & 
		\end{array} \right.
	\end{equation}
	for all $y \in \Sigma$, with $\delta_y$ the Dirac delta function and with $g$ smooth. In the open case we will consider Neumann boundary conditions, in which case Green's functions are solutions of
	\begin{equation}\label{eq:Neumann_pb_green}
		\left\lbrace \begin{array}{ll}
			-\Delta_g G_g(\cdot,y) = 2\pi\left(\delta_y - \frac{1}{v_g(\Sigma)}\right) & \text{in } \Sigma \\
			\partial_{\vec{n}_g} G_g(\cdot,y) = 0 & \text{on } \partial\Sigma\\
			\int_\Sigma G_g(x,y)\d v_g(x) = 0. &
		\end{array} \right.
	\end{equation}
	These Green's functions are such that for any $f$ smooth over $\overline\Sigma$ and with $m_g(f)\coloneqq \frac1{v_g(\Sigma)}\int_\Sigma f\d v_g$
	\begin{equation} \label{eq:Green}
		\frac1{2\pi}\int_{\partial\Sigma} G_g(x,\cdot)\partial_{\vec{n}_{g}} f \d \lambda_g + \frac1{2\pi}\int_{\Sigma}  G_g(x,\cdot)\left(-\Delta_gf\right)  \d v_g = f(x) - m_g(f).
	\end{equation}
	Moreover they have a logarithmic divergence on the diagonal in the sense that:
	\begin{lemma}\label{lemma:approx_green}
		There exists $W$ (resp. $W_\partial$) continuous in $\Sigma$ (resp. on $\partial\Sigma$) such that, as $x\to y$,
		\begin{eqs}\label{eq:div_green}
			&G_{g}(x,y) = -\log d_{g}(x,y) + W(y)+o(1)\qt{if}y\in\Sigma\qt{and}\\ &G_{g}(x,y) = -2\log d_{g}(x,y) +W_\partial(y)+o(1)\qt{if}y\in\partial\Sigma.
		\end{eqs}
		In the rest we will denote $W_\partial=W$ to keep the notations not too heavy.
	\end{lemma}
	\begin{proof}
		If $\Sigma$ is closed this statement follows from~\cite[Lemma 3.2]{GRV16} in the case of negative Euler characteristic, \cite[Equation (3.4)]{DRV16} for $\chi(\Sigma)=0$ and~\cite[Proposition 2.5]{DKRV} for positive Euler characteristic. If $\Sigma$ has non-empty boundary thanks to uniformization  it suffices to show the statement for a uniform metric of type $I$. For such a metric we can use the \lq\lq doubling trick" to deduce the statement from the closed one, see \textit{e.g.} the proof of~\cite[Lemma 3.3]{CH_construction}.
	\end{proof}
	This property allows to write down the conical singularities in terms of a singular conformal factor. Namely given a divisor $\Div$ as above we now introduce the function over $\overline\Sigma$:
	\begin{equation}
		\Hdiv(x)\coloneqq -2\sum_{k=1}^{N+M}a_k G_g(x,z_k).
	\end{equation}
	Then for any $p<p*\coloneqq \min\limits_{a_k<0}\frac{1}{\norm{a_k}}\wedge\min\limits_{b_l<0}\frac{1}{\norm{b_l}}$, $e^{\Hdiv}$ belongs to $L^{p}_{\text{loc}}(\Sigma,g)$ while $e^{\frac12\Hdiv}$ is an element of $L^p_{\text{loc}}(\partial\Sigma,g)$. 
	Moreover an admissible metric can be represented as $e^{\varphi+\Hdiv}g$, where $g$ is a smooth Riemannian metric on $\overline{\Sigma}$, while $\varphi$ is continuous over $\overline\Sigma$. We note that via the uniformization for open Riemann surfaces $g$ can be taken to be a uniform metric of type $I$.
	
	\subsubsection{Moser-Trudinger inequality}
	Following~\cite{Troyanov}, we define the \textit{Trudinger constant} $\tdiv$ associated to the surface $\Sigma$ and the divisor $\Div$ by setting
	\begin{equation}
		\tdiv\coloneqq 1\wedge \min_{1\leq k\leq N+M} 1+a_k.
	\end{equation}
	It naturally appears in the curvature prescription problem especially thanks to the following Moser-Trudinger inequalities proved in~\cite[Lemmas 3.2 and 3.3]{BRS}:
	\begin{proposition}\label{prop:moser_trudinger}
		Assume that $\tau'<\tdiv$ and that $g$ is smooth. Then there exists a positive constant $C$ such that for any $\varphi\in \dot H^1(\Sigma,g)$:
		\begin{eqs}
			4\tau'\ln\left(\int _{\partial\Sigma}e^{\frac12\left(\varphi+\Hdiv\right)}\d l_{g}\right)
			&\leq C+\frac{1}{4\pi}\int_\Sigma\norm{\nabla_g\varphi}_g^2dv_g\\
			8\tau'\ln\left(\int _{\Sigma}e^{\varphi+\Hdiv}\d v_{g}\right)
			&\leq C+\frac{1}{4\pi}\int_\Sigma\norm{\nabla_g\varphi}_g^2dv_g.
		\end{eqs}
	\end{proposition}
	In the above statement $H^1(\Sigma,g)=W^{1,2}(\Sigma,g)$ is the standard Sobolev space while \\
	${\dot{H}^1(\Sigma,g)\coloneqq\left\{\varphi\in H^1(\Sigma,g),\text{ } m_g(\varphi)=0\right\}}$. Thanks to these inequalities we have the following:
	\begin{lemma}\label{lemma:cont}
		Take $p>1$. Then for $\Lambda$ in $L^p(\Sigma,e^{\varphi+\Hdiv}g)$ and $\sigma$ in $L^p(\partial\Sigma, e^{\varphi+\Hdiv}g)$ the map
		\[
		\varphi\mapsto \int_\Sigma \Lambda e^{\varphi+\Hdiv}\d v_{g}+\int_{\partial\Sigma}4\sigma e^{\frac12\left(\varphi+\Hdiv\right)}\d l_{g}
		\]
		is weakly continuous over $H^1(\Sigma,g)$.
	\end{lemma}
	\begin{proof}
		Let $\left(\varphi_n\right)_{n\in\N}$ converge weakly towards $\varphi$ in $H^1(\Sigma,g)$. Then using Proposition~\ref{prop:moser_trudinger}, for any $q<\infty$ the sequence in $L^q(\Sigma,g)$ given by $\left(e^{\varphi_n}\right)_{n\in\N}$ converges to $e^{\varphi}$ for the $L^q(\Sigma,\gdiv)$ norm. The same applies on the boundary: $\left(e^{\frac12\varphi_n}\right)_{n\in\N}$ converges strongly to $e^{\frac12\varphi}$ in $L^q(\partial\Sigma,\gdiv)$. By continuity of the $L^p\times L^q$ pairing for $\frac1q+\frac1p=1$ we get the result.
	\end{proof}
	
	
	
	\subsection{Prescription of scalar and geodesic curvatures and the Liouville action}
	We now investigate a problem analogous to the uniformization of open Riemann surfaces in the presence of conical singularities. To this end let $\Div$ be a divisor as above, $\Lambda$ and $\sigma$ be continuous respectively on $\Sigma$ and on $\partial\Sigma\setminus\bm z$. 
	We further assume that both $\Lambda$ and $\sigma$ are non-negative and bounded, and that $\Lambda\not\equiv 0$. We are interested in the following question:
	\begin{prob}\label{prob:class}
		Let $(\Sigma,g)$ be a smooth Riemannian surface with boundary. Find a Riemannian metric $\gdiv=e^{\phidiv}g$ on $\Sigma$ representing $\Div$ and with:
		\begin{itemize}
			\item $\gdiv=e^{\varphi+H_{\bm z,\bm a}}g$ with $\varphi$ continuous over $\overline\Sigma$; 
			\item Gaussian curvature $-\frac12\Lambda$ inside $\Sigma$;
			\item geodesic curvature $-\sigma$ on $\partial\Sigma\setminus\bm z$.
		\end{itemize}
	\end{prob}
	
	\subsubsection{Existence and uniqueness of solutions to Problem~\ref{prob:class}}
	This problem can be reformulated in terms of a regular function $\varphi$ over $\overline\Sigma$. Namely without loss of generality take $g$ a uniform metric of type $I$ and consider $\Hdiv$ as above. We want to find a solution $\varphi$ of
	\begin{equation}\label{eq:curv_problem}
		\left\lbrace \begin{array}{ll}
			\Delta_g \varphi = R_{g} +\Lambda e^{\varphi+\Hdiv} & \text{in } \Sigma \\
			\partial_{\vec{n}_g}\varphi = -2\sigma e^{\frac12\left(\varphi+\Hdiv\right)} & \text{on }\partial\Sigma.
		\end{array} \right.
	\end{equation}
	Integrating the above because of the Gauss-Bonnet formula a necessary condition for this problem to have a solution is $\eul<0$. It is actually a sufficient one:
	\begin{theorem}\label{thm:prescribe}
		Assume that $\eul<0$. Then Problem~\ref{prob:class} admits a unique solution. Moreover if $\Lambda$ is constant while $\sigma$ is piecewise constant with possible discontinuities over $\bm z$ then $\varphi$ is smooth inside $\Sigma$ and over each connected component of $\partial\Sigma\setminus\bm z$.
	\end{theorem}

	More generally as discussed in the introduction one can remove the assumption that the curvatures are negative, or relax the hypothesis that the singular Euler characteristic is negative (see e.g.~\cite{LSMR, BRS}). We however consider here only the setting where $\eul<0$ for which elementary tools allow to show existence and uniqueness of solutions for Problem~\ref{prob:class}. 
	
	\begin{proof}
		To be more specific and motivated by Liouville CFT we use a variational approach based on the Liouville action~\eqref{eq:action} and the definition of the correlation functions. Indeed they formally lead to the consideration, for suitable $\phi$, of
		\begin{eqs}\label{eq:action_za}
			S(\phi;\Div) \lq\lq =" &\frac1{4\pi}\int_\Sigma\left(\norm{\nabla_g \phi}_g^2+2R_g\phi + 2\Lambda e^{\phi}\right)\d v_g+\frac1{2\pi}\int_{\partial\Sigma}\left(2k_g\phi + 4\sigma e^{\phi/2}\right)\d l_g\\
			&+2\sum_{k=1}^Na_k\phi(z_k)+\sum_{l=1}^Mb_l\phi(z_l).
		\end{eqs}
		Due to the singular behaviour of the field $\phi$ some care is required to make sense of the latter: we explain how to address this issue in the following Subsection. For the moment let us assume without loss of generality that $g$ is a uniform metric of type $1$ and write $\phi=\varphi+\Hdiv$ with $\varphi$ in $H^{1}(\Sigma,g)$. We are thus interested in the functional over $H^{1}(\Sigma,g)$
		\begin{equation}
			I_{\bm z,\bm \alpha}(\varphi)=\frac1{4\pi}\int_\Sigma\left(\norm{\nabla_g \varphi}_g^2+ 2\Lambda e^{\varphi+\Hdiv}\right)\d v_g+\frac2{\pi}\int_{\partial\Sigma}\sigma e^{\frac12\left(\varphi+\Hdiv\right)}\d l_g+2\eul c
		\end{equation}
		with $c=m_g(\varphi)$.
		Using the Moser-Trudinger inequality from Proposition~\ref{prop:moser_trudinger} this action is indeed well-defined for $\varphi$ in $H^{1}(\Sigma,g)$. Moreover because of the assumptions that $\Lambda$ and $\sigma$ are non-negative with $\Lambda>0$ on some (non-empty) open subset, the condition that $\eul<0$ implies that $I_{\bm z,\bm\alpha}$ is coercive. Moreover it is seen to be weakly lower semicontinuous thanks to Lemma~\ref{lemma:cont} above, as well as strictly convex. As a consequence $I_{\bm z,\bm\alpha}$ admits a unique critical point $\varphi_{\bm z,\bm a}$ in $H^1(\Sigma,g)$ which is thus a minimum.
		
		Now $\varphi_{\bm z,\bm a}$ being a critical point of $I_{\bm z,\bm\alpha}$ amounts to being a weak solution of Equation~\eqref{eq:curv_problem}. As a consequence we have shown existence and uniqueness of a weak solution for our problem. Continuity of $\varphi_{\bm z,\bm a}$ follows from $L^p$ regularity for the Laplacian (our boundary is smooth and $\Lambda e^{\varphi_{\bm z,\bm a}+\Hdiv}$, $\sigma e^{\frac12\left(\varphi_{\bm z,\bm a}+\Hdiv\right)}$ meet the requirement of~\cite[Theorem 2.5.1.1]{Gris}), from which we can deduce smoothness by recursive application of H\"older regularity~\cite[Theorem 6.3.1.4]{Gris} if in addition $\Lambda$ and $\sigma$ are smooth. This concludes the proof of Theorem~\ref{thm:prescribe}.
	\end{proof}

	\begin{remark}
		In analogy with the definition of the unit boundary length quantum disk we can further fix the value of $c$ and end up with the functional valued over $\dot{H}^{1}(\Sigma,g)$:
		\[
		J(\varphi)\coloneqq \frac1{4\pi}\int_\Sigma\norm{\nabla_g \varphi}_g^2\d v_g+\frac{\frac1{2\pi}\int_\Sigma \Lambda e^{\varphi+\Hdiv}\d v_{g}}{\left(\frac2{\pi}\int_{\partial\Sigma}\sigma e^{\frac12\left(\varphi+\Hdiv\right)}\d l_{g}\right)^2}-2\eul\ln \frac2{\pi}\int_{\partial\Sigma}\sigma e^{\frac12\left(\varphi+\Hdiv\right)}\d l_{g}.
		\]
	\end{remark}
	In the rest we will assume that $\Lambda>0$ and $\sigma\geq 0$ are respectively constant and piecewise constant with possible discontinuities over $\bm z$ and study the associated Liouville action.
	
	\subsubsection{The classical Liouville action}
	Let us denote by $\varphi_{\bm z,\bm a}$ the solution of Problem~\ref{prob:class} given by Theorem~\ref{thm:prescribe} and set $\phidiv=\varphi_{\bm z,\bm a}+\Hdiv$: we give here a meaning to $S_{\bm z,\bm \alpha}(\phidiv)$, where without loss of generality we assume that $g$ is uniform of type $I$. 
	For this we will use a limiting procedure by first regularizing $\phidiv$ by averaging it over geodesic (semi-) circles:
	\begin{equation}
		\phi_\rho(x)\coloneqq \frac1{l_g(\partial B(x,\rho))}\int_{\partial B(x,\rho)}\phidiv\d l_g.
	\end{equation}
	
	We are then interested in the limit as $\rho,\eps$ and then $\delta\to0$ of $S_{\delta,\eps}(\phi_\rho)$ where $S_{\delta,\eps}(\phi_\rho)$ is defined as in Equation~\eqref{eq:action_za} with $\Sigma$ (resp. $\partial\Sigma$) being replaced by $\Sigma_{\delta,\eps}$ (resp. $\left(\partial\Sigma\right)_\eps$):
	\begin{eqs}\label{eq:action_reg}
		S_{\delta,\eps}(\phi)\coloneqq &\frac1{4\pi}\int_{\Sigma_{\delta,\eps}}\left(\norm{\nabla_g \phi}_g^2+2R_g\phi + 2\Lambda e^{\phi}\right)\d v_g+\frac1{2\pi}\int_{\left(\partial\Sigma\right)_\eps}\left(2k_g\phi + 4\sigma e^{\phi/2}\right)\d l_g\\
		&+2\sum_{k=1}^Na_k\phi(z_k)+\sum_{l=1}^Mb_l\phi(z_l).
	\end{eqs}
	\begin{proposition}\label{prop:def_action}
		The following limit exists and is finite as $\rho,\eps$ and then $\delta\to0$:
		\begin{equation}
			S(\phidiv;\Div)\coloneqq \lim\limits_{\delta,\eps,\rho\to0}S_{\delta,\eps}\left(\phi_\rho\right)-\left(2\sum_{k=1}^Na_k^2+\sum_{l=1}^Mb_l^2\right)\ln\rho.
		\end{equation}
		Moreover $S(\phidiv;\Div)=I_{\bm z,\bm a}(\vphidiv)+G(\bm z,\bm a)$ where, with $W$ the function from Lemma~\ref{lemma:approx_green},
		\begin{equation}
			G(\bm z,\bm a)\coloneqq -2\sum_{k\neq l}a_ka_lG_g(z_k,z_l)-2\sum_{k=1}^{N+M}a_k^2W(z_k)
		\end{equation}
	\end{proposition}
	\begin{proof}
		We write $\phi_\rho=\varphi_\rho+H_\rho$ with $\varphi_\rho=\varphi_{\bm z,\bm a}*\eta_\rho$ and $H_\rho=\Hdiv*\eta_\rho$. Then using Equation~\eqref{eq:Green} we have up to a term that vanishes in the $\rho,\eps,\delta\to0$ limit,
		\begin{align*}
			S_{\delta,\eps}(\phi_\rho)=I_{\bm z,\bm a}(\varphi_{\bm z,\bm a})+&\frac1{2\pi}\int_{\Sigma_{\delta,\eps}}\nabla_g \varphi_{\bm z,\bm a}\cdot\nabla_g \Hdiv\d v_g+2\sum_{k=1}^{N+M}a_k\varphi_{\bm z,\bm a}(z_k)\\
			+&\frac1{4\pi}\int_{\Sigma_{\delta,\eps}}\norm{\nabla_g H_\rho}_g^2\d v_g+2\sum_{k=1}^{N+M}a_kH_\rho(z_k)+o(1).
		\end{align*}
		We treat the second and third lines by using the fact that for $\psi$ smooth over $\Sigma_{\delta,\eps}$:
		\begin{align*}
			\int_{\Sigma_{\delta,\eps}}\nabla_g \psi\cdot\nabla_g \Hdiv\d v_g=\int_{\partial\left(\Sigma_{\delta,\eps}\right)}\psi\partial_{n_g} \Hdiv\d l_g=\sum_{z\in\bm z}\int_{\partial B(z,\eps)}\psi\partial_{n_g} \Hdiv\d l_g+o(1).
		\end{align*}
		This shows that the second line vanishes in the limit. As for the third line it is given by
		\begin{align*}
			-2\sum_{k\neq l}a_ka_l G_g(z_k,z_l)-2\sum_{k=1}^{N+M}a_k^2\left(2G_\rho(z_k,z_k)-\left(G_\rho\right)_\eps(z_k,z_k)\right)+o(1).
		\end{align*}
		The proof is concluded thanks Lemma~\ref{lemma:approx_green}. 
	\end{proof}
	We denote by $S_{\bm z,\bm a}\coloneqq S(\phidiv;\Div)$ the corresponding Liouville action defined by Proposition~\ref{prop:def_action}. Holomorphic derivatives of $S_{\bm z,\bm a}$ with respect to the punctures describe the accessory parameters. However defining them properly requires some care since such derivatives may actually be ill-defined; we will address this issue in Section~\ref{sec:accessory} thanks to the semi-classical limit of Liouville theory and based on a regularization of $\phidiv$ and $\Sigma$.

	
	
	\subsection{Regularization of the Liouville field and action}
	Before actually defining the derivatives of $\phidiv$ as well as the classical stress-energy tensor we describe here a limiting procedure for defining $\phidiv$ and $S_{\bm z,\bm a}$, key for the purpose of the next section. 
	
	\subsubsection{Regularization of $\phidiv$}\label{subsec:phireg}
	To start with we want to regularize the Liouville action and the Liouville field $\phidiv$. The way we do so is by changing the underlying surface $\Sigma$ in order to avoid the singular points. A way to achieve that is to consider for $\eps,\delta$ positive $\Sigma_{\delta,\eps}$ and $\left(\partial\Sigma\right)_\eps$ defined like before for punctures at $\bm z$ and let $\phireg$ be the (unique) solution of Problem~\ref{prob:class} with the curvatures given by $\Lambda_{\delta,\eps}$ and $\sigma_{\eps}$ where the latter are smooth, non-negative and such that:
	\begin{align*}
		\Lambda_{\delta,\eps}&\equiv\Lambda\text{ inside }\Sigma_{\delta,\eps}\qt{and}\Lambda_{\delta,\eps}(x)\equiv0\text{ if }\min_{z\in\bm z}d_g(x,z)<\eps-\eps^{10}\text{ or }d_g(x,\partial\Sigma)<\delta-\delta^{10},\\
		\sigma_{\eps}&\equiv\sigma\text{ on }\left(\partial\Sigma\right)_\eps\qt{and}\sigma_{\eps}(x)\equiv0\text{ if }\min_{z\in\bm z} d_g(x,z)<\eps-\eps^{10}.
	\end{align*}
	Then $\vphireg\coloneqq\phireg-\Hdiv$ is smooth over $\overline\Sigma$, and in particular near the punctures. 
	We can also approximate $S_{\bm z,\bm a}$ by considering the Liouville action associated to the curvatures $\Lambda_{\delta,\eps}$, $\sigma_\eps$ and evaluated at the field $\phireg$. That is to say we study the $\delta,\eps\to0$ limit of
	\begin{equation}
		S_{\delta,\eps}\coloneqq \lim\limits_{\rho\to0} S_{\delta,\eps}\left((\phireg)_\rho\right)-\left(2\sum_{k=1}^Na_k^2+\sum_{l=1}^Mb_l^2\right)\ln\rho.
	\end{equation}
	
	\subsubsection{Approximation of $\phidiv$ and $S_{\bm z,\bm a}$}
	Before moving on we show that the above indeed provides good approximations of the field $\phidiv$ and of the Liouville action $S_{\bm z,\bm a}$.
	\begin{lemma}\label{lemma:approx}
		In the limit where $\eps$ and then $\delta$ go to $0$, $S_{\bm z,\bm a}=\lim\limits_{\delta,\eps\to0} S_{\delta,\eps}$ and $\vphireg$ converges in $H^{1}(\Sigma,g)$ to $\vphidiv$. In particular for $F$ (resp. $G$) in $L^q(\Sigma,g)$ (resp. $L^{q}(\partial\Sigma,g)$) with $q>\frac{p*}{p*-1}$
		\begin{equation}
			\lim\limits_{\delta,\eps\to0}\int_{\Sigma} F e^{\phireg}\d v_g+\int_{\partial\Sigma}Ge^{\frac12\phireg}\d l_g=	\int_{\Sigma} F e^{\phidiv}\d v_g+\int_{\partial\Sigma}Ge^{\frac12\phidiv}\d l_g.	
		\end{equation}
	\end{lemma}
	\begin{proof}
		Since $\vphireg$ minimizes the (coercive) Liouville action $\nnorm{\vphireg}_{H^1}$ is uniformly bounded in $\delta,\eps$ small enough. Moreover $\Lambda_{\delta,\eps}$ converges to $\Lambda$ in $L^p(\Sigma,g)$ for any $p<\infty$ (and likewise for $\sigma_\eps$): as a consequence thanks to Proposition~\ref{prop:moser_trudinger} for any $\rho>0$ we have $S_{\bm z,\bm a}(\phireg)\leq S_{\delta,\eps}(\phireg)+\rho$ for $\delta,\eps$ small enough. Since $\phidiv$ minimizes the Liouville action we get $S_{\bm z,\bm a}\leq \liminf\limits_{\delta,\eps\to0} S_{\delta,\eps}$. The same argument shows that $S_{\bm z,\bm a}\geq \limsup\limits_{\delta,\eps\to0} S_{\delta,\eps}$: we conclude that $S_{\bm z,\bm a}= \lim\limits_{\delta,\eps\to0} S_{\delta,\eps}$.
		Now 
		\begin{align*}
			I_{\bm z,\bm a}(\vphidiv+u)&=I_{\bm z,\bm a}(\vphidiv)+\frac{1}{4\pi}\int_\Sigma\norm{\nabla_gu}_g^2\d v_g\\
			&+\frac{1}{2\pi}\int_\Sigma\left(e^u-1-u\right)\Lambda e^{\phidiv}\d v_g+\frac{2}{\pi}\int_{\partial\Sigma}\left(e^{\frac12u}-1-\frac12u\right)\sigma e^{\frac12\phidiv}\d l_g.
		\end{align*}
		for any $u$ in $H^1$, and in particular for $u=u_{\delta,\eps}=\vphidiv-\vphireg$. Since $e^u\geq 1+u$ the equality $S_{\bm z,\bm a}= \lim\limits_{\delta,\eps\to0} S_{\delta,\eps}$ implies that $\nnorm{u_{\delta,\eps}}_{H^1}=\nnorm{\vphidiv-\vphireg}_{H^1}\to0$.  
		The last point then follows from Lemma~\ref{lemma:cont} since $F\in L^{p}_{\text{loc}}(\Sigma,\gdiv)$ for any $p<q\frac{p_*-1}{p_*}$ (the latter being strictly greater than $1$) and that $e^{\vphireg}$ converges to $e^{\vphidiv}$ in $L^{r}(\Sigma,\gdiv)$ for any $r<\infty$ via Proposition~\ref{prop:moser_trudinger}.
	\end{proof}
	We also study the dependence of the Liouville action with respect to the punctures: 
	\begin{lemma}
		Viewed as a function of the bulk punctures $\left(z_1\mapsto S_{\delta,\eps}\right)$ converges uniformly to $\left(z_1\mapsto S_{\bm z,\bm a}\right)$ on every compact of $\H\setminus\{z_2,\cdots,z_N\}$. The same applies for a boundary insertion. 
	\end{lemma}
	\begin{proof}
		By the same argument as in the proof of Lemma~\ref{lemma:approx} it suffices to show that for any $K\subseteq \H\setminus\{z_2,\cdots,z_N\}$ we have a bound of the form
		$\sup\limits_{z_1\in K}\nnorm{\vphireg^{(z_1)}-\vphidiv^{(z_1)}}_{H^1}\leq r_{\delta,\eps}$ for some $r_{\delta,\eps}\to0$, and where the exponent $(z_1)$ indicates the dependence in the puncture $z_1$. To start with for $\eps$ small enough we have $K\cap\bm z_\eps=\emptyset$. Now since $\vphireg^{(z_1)}-\vphidiv^{(z_1)}$ is supported in $\Sigma\setminus \Sigma_{\delta,\eps}$ by Cauchy-Scwharz inequality it suffices to show that $\sup\limits_{z_1\in K}\sup\limits_{\delta,\eps}\nnorm{\vphireg^{(z_1)}}_{H^1}<\infty$ which is readily seen as soon as $\eps$ and $\delta$ are small enough.
	\end{proof}
	We stress that $S_{\delta,\eps}$ corresponds to the Liouville action $S(\phidiv,\Div)$ associated to the divisor $\Div$ but with curvatures given by $\Lambda_{\delta,\eps}$ and $\sigma_{\delta,\eps}$ instead of $\Lambda$ and $\sigma$. This approximation will turn out to be fundamental in order to make sense of the stress-energy tensor as well as the accessory parameters in the next section.
	This is due to the following formula:
	\begin{equation}\label{eq:key_formula1}
		\phidiv(x)=\Hdiv+c-\frac1{2\pi}\int_{\Sigma} G_g(x,y)\Lambda e^{\phidiv(y)}\d v_g - \frac1{\pi}\int_{\partial\Sigma} G_g(x,y)\sigma e^{\frac12\phidiv(y)}\d l_g
	\end{equation}
	which is a consequence of Equation~\eqref{eq:Green} since $\phidiv=\vphidiv+\Hdiv$  where $\vphidiv$ solves Equation~\eqref{eq:curv_problem}.
	Despite its elementary form this formula is fundamental in the study of the dependence of the Liouville action in the location of the punctures as we demonstrate in the next section. To this end we will apply it to the regularized field $\phireg$, that is use
	\begin{equation}\label{eq:key_formula}
		\phireg(x)=\Hdiv+c-\frac1{2\pi}\int_\Sigma G_g(x,y)\Lambda_{\delta,\eps} e^{\phireg(y)}\d v_g - \frac1{\pi}\int_{\partial\Sigma} G_g(x,y)\sigma_{\delta,\eps}e^{\frac12\phireg(y)}\d l_g.
	\end{equation}

	
	
	\section{Stress-energy tensor and accessory parameters on $\H$}\label{sec:SET}
	In this section we assume the underlying Riemann surface to be the upper half-plane $\H$. We fix a divisor $\Div$ and consider $\vphidiv$ the unique solution (see below for a justification) of \begin{equation}
		\left\lbrace \begin{array}{ll}
			-\Delta \vphidiv= \Lambda e^{\vphidiv+\Hdiv} & \text{in } \H \\
			\partial_{n}\vphidiv= -2\sigma e^{\frac12\left(\vphidiv+\Hdiv\right)} & \text{on }\R.
		\end{array} \right.
	\end{equation}
	We define in this section the classical \textit{stress-energy tensor} $T(z)$ in terms of $\Phidiv\coloneqq\vphidiv+\Hdiv$, and show that it can be expressed in terms of \textit{accessory parameters}:
	\begin{theorem}\label{thm:SET}
		Let $T(z)\coloneqq \partial_z^2\Phidiv(z)-\frac12\left(\partial_z\Phidiv(z)\right)^2$ for $z\in\H\setminus\bm z$, while for $t\in\R\setminus\bm z$
		\[
		T(t)\coloneqq-\frac\Lambda4 e^{\Phidiv(t)}+\lim\limits_{\delta,\eps\to0}\left(\frac12\partial_t^2\Phireg(t)-\frac18\left(\partial_t\Phireg(t)\right)^2+\frac{1}{\eps}\frac{\sigma(t)}{2\pi}e^{\frac12\Phireg(t)}\right)
		\]
		where $\Phireg$ is a regularization of $\Phidiv$ (see below). Then for $x\in\overline \H\setminus\bm z$
		\begin{equation}
			T(x)=\sum_{k=1}^N\left(\frac{\delta_k}{(x-x_k)^2}+\frac{\bm c_k}{x-x_k}+\frac{\delta_k}{(x-\overline{x_k})^2}+\frac{\overline{\bm c_k}}{x-\overline{x_k}}\right)+\sum_{l=1}^M\left(\frac{\delta_l}{(x-s_l)^2}+\frac{\bm c_l}{x-s_l}\right)
		\end{equation}
		where the weights are $\delta_k=-a_k(1+\frac12{a_k})$ and the accessory parameters $\bm c_k$ are explicit.
	\end{theorem}
	We will prove in the Proposition~\ref{prop:accessory} that $\bm c_k=-\frac12\partial_{z_k}S_{\bm z,\bm a}$. In the present section we give an explicit expression for them in terms of the field $\Phidiv$ in Equations~\eqref{eq:weight_accessory} and~\eqref{eq:L1_explicit}. We also actually prove here a more general statement that formally corresponds to the evaluation of the stress-energy tensor at singular points:
	\begin{theorem}\label{thm:L2}
		Set $\Lc_{-2}^{(k)}[\Phidiv] \coloneqq \left(1-a_k\right)\partial^2\Phidiv^{(k)}(z_k)-\frac12\left(\partial\Phidiv^{(k)}(z_k)\right)^2$ for $1\leq k\leq N$ and
		\[
		\Lc_{-2}^{(l)}[\Phidiv] \coloneqq \lim\limits_{\delta,\eps\to0}\left(\frac12\left(1-b_l\right)\partial^2\Phireg^{(k)}(s_l)-\frac18\left(\partial\Phireg^{(l)}(s_l)\right)^2-\mathfrak R_{\delta,\eps}^{(l)}\right)
		\]
		with $\mathfrak R_{\delta,\eps}^{(k)}$ an explicit remainder term (see Lemma~\ref{lemma:L2}) and $\Phi^{(k)}=\Phi-2a_k\ln\norm{\cdot-z_k}$. Then 
		\begin{eqs}
			\Lc_{-2}^{(l)}[\Phidiv]&=\sum_{\substack{1\leq k\leq M\\k\neq l}}\left(\frac{\delta_l}{(s_l-s_k)^2}+\frac{\bm c_l}{s_l-s_k}\right)\\
			&+\sum_{k=1}^N\left(\frac{\delta_k}{(s_l-x_k)^2}+\frac{\bm c_k}{s_l-x_k}+\frac{\delta_k}{(s_l-\overline{x_k})^2}+\frac{\overline{\bm c_k}}{s_l-\overline{x_k}}\right).
		\end{eqs}
	\end{theorem}
	\begin{remark}\label{rmk:BPZ}
		If we take only boundary insertions with weights $b_l=1$ for all $l$ (note however that if we do so we cannot have $\eul<0$) then thanks to Equation~\eqref{eq:L2_expr} we have $\Lc_{-2}^{(l)}[\Phidiv] =-\frac18\left(\partial\Phidiv^{(l)}(s_l)\right)^2+\frac{\sigma(s_l^-)+\sigma(s_l^+)}{2\pi} e^{\frac12\Phidiv^{(l)}(s_l)}$. Hence 
		\begin{equation*}
			\frac12\left(\partial_{s_l}\Phidiv^{(l)}\right)^2+2\frac{\sigma(s_l^-)+\sigma(s_l^+)}{\pi} e^{\frac12\Phidiv^{(l)}(s_l)}+\sum_{k\neq l}\frac{4\bm c_l}{s_l-s_k}=\sum_{k\neq l}\frac{6}{(s_l-s_k)^2}\cdot
		\end{equation*}
		Anticipating on Proposition~\ref{prop:accessory} we thus see that 
		\begin{equation}\label{eq:HEM_heavy}
			\frac12\left(\partial_{s_l}S_{\bm z,\bm a}\right)^2-\sum_{k\neq l}\frac{2}{s_l-s_k}\partial_{s_k}S_{\bm z,\bm a}=\sum_{k\neq l}\frac{6}{(s_l-s_k)^2}-2\frac{\sigma(s_l^-)+\sigma(s_l^+)}{\pi} e^{\frac12\Phidiv^{(l)}(s_l)}.
		\end{equation}
		In particular if $\sigma(s_l^-)+\sigma(s_l^+)=0$ for all $l$ this coincides with the differential equation satisfied by the minimal (multichordal) Loewner potential from~\cite[Proposition 1.8]{PeWa}.
	\end{remark}
	Though the statements have been formulated on $\H$ they are actually purely local and a similar result should hold on more general Riemann surfaces. See \textit{e.g.}~\cite{BaWu} where this aspect is discussed in relation with the Hilbert space picture for Liouville CFT.
	
	
	
	\subsection{Preliminary computations}
	Let $\Sigma=\H$ be the upper half-plane and $\partial\Sigma=\R$ be its boundary. We equip it with the Riemannian metric $g_0=\frac{4\norm{dx}^2}{(1+\norm{x}^2)^2}=: e^{2w_0(x)}\norm{dx}^2$ so that we are in the setting of Theorem~\ref{thm:prescribe}.  This metric is geodesic ($k_{g_0}=0$) and has constant scalar curvature $R_{g_0}=2$.   
	The corresponding Green's function is given by
	\begin{equation}\label{eq:green_H}
		G_0(x,y)=\ln\frac{1}{\norm{x-y}\norm{x-\bar y}}-\left(w_0(x)+w_0(y)+1\right).
	\end{equation}
	
	\subsubsection{Definition of the fields}
	Let $\phidiv$ be the solution of Problem~\ref{prob:class} with $\Lambda>0$ and $\sigma\geq 0$, constant respectively in $\H$ and on each connected component of $\R\setminus\bm z$. We also introduce $\Phidiv\coloneqq\phidiv+2w_0$ and $\vphidiv\coloneqq \Phidiv-\Hdiv$, which is thus a solution of
	\begin{equation}
		\left\lbrace \begin{array}{ll}
			-\Delta \vphidiv= \Lambda e^{\vphidiv+\Hdiv} & \text{in } \H \\
			\partial_{n}\vphidiv= -2\sigma e^{\frac12\left(\vphidiv+\Hdiv\right)} & \text{on }\R.
		\end{array} \right.
	\end{equation}
	It is not hard to see that the Liouville action associated to $\phidiv$ is given by
	\begin{equation}
		S_{\bm z,\bm a}=c+G(\bm z,\bm a)+\frac1{4\pi}\int_\H\left(\norm{\nabla\vphidiv}^2+2\Lambda e^{\Phidiv}\right)\d^2x+\frac2{\pi}\int_\R\sigma e^{\frac12\Phidiv}\d t
	\end{equation}
	where $c$ only depends on $w_0$. Moreover as a consequence of Equation~\eqref{eq:key_formula1} we have:
	\begin{equation}\label{eq:key_formula2}
		\Phidiv=\Hdiv+2w_0+c-\frac1{2\pi}\int_{\H} G_0(\cdot,y)\Lambda e^{\Phidiv(x)}\d^2x - \frac1{\pi}\int_{\R} G_0(\cdot,t)\sigma e^{\frac12\Phidiv(t)}\d t.
	\end{equation}

	Let us fix any positive $r$ with  $2r<\min\limits_{z\neq z'\in\bm z}\norm{z-z'}\wedge\min\limits_{z\in \bm z\cap\H}\Im(z)$ and pick $z_0\in(\H+\im r)\setminus\bm z_r$ and $s_{M+1}\in\R\setminus\bm z_r$.
	We define $\Phidiv^{(k)}\coloneqq \Phidiv-2a_k\ln\norm{\cdot-z_k}$ for $0\leq k\leq N$ and likewise set $\Phidiv^{(l)}$ for a boundary insertion (there will be no ambiguity in the label). We consider $\Lambda_{\delta,\eps}$ and $\sigma_{\delta,\eps}$ defined like in Subsection~\ref{subsec:phireg} for the singular points $\bm z\cup\{z_0,s_{M+1}\}$: this yields a regularized field $\phireg$ and $\Phireg\coloneqq\phireg+2w_0$. We have the analog of Equation~\eqref{eq:key_formula2}:
	\begin{equation}\label{eq:key_formula2_reg}
		\Phireg=\Hdiv+2w_0+c-\frac1{2\pi}\int_{\H} G_0(\cdot,x)\Lambda_{\delta,\eps} e^{\Phireg(x)}\d^2x - \frac1{\pi}\int_{\R} G_0(\cdot,t)\sigma_\eps e^{\frac12\Phireg(t)}\d t.
	\end{equation}
	We will also work with $\Phireg^{(k)}\coloneqq \Phireg-2a_k\ln\norm{\cdot-z_k}$ for boundary and bulk insertions.

	\subsubsection{Derivatives of $\phireg$ and a fundamental formula}
	For $z$ on $\H$ let us denote by $\partial=\partial_z$ the holomorphic (Wirtinger) derivative and for $t\in\R$ the usual (real) derivative. Because the classical field $\Phidiv$ may not be differentiable we need to go through a limiting procedure involving its regularized counterpart $\Phireg$. To do so in analogy with boundary Liouville theory we set for $0\leq k\leq N$ and $1\leq l\leq M+1$:
	\begin{equation}
		\Lc_{-1}^{(k)}[\Phireg]\coloneqq -2a_k\partial\Phireg^{(k)}(z_k)\qt{and}\Lc_{-1}^{(l)}[\Phireg]\coloneqq -b_l\partial\Phireg^{(l)}(s_l)
	\end{equation}
	by which we mean $-2 a_k\partial_x\Phireg(x)$ evaluated at $x=z_k$ (and likewise for $s_l$).
	\begin{lemma}\label{lemma:formula_der}
		For any positive $\delta,\eps$ and $0\leq k\leq N$ we have:
		\begin{eqs}\label{eq:L1}
			&\Lc_{-1}^{(k)}[\Phireg]=\frac{2a_k^2}{\overline z_k-z_k}+\sum_{\substack{0\leq l\leq N\\ l\neq k}}\frac{2a_ka_l}{z_l-z_k}+\frac{2a_ka_l}{\overline z_l-z_k}+\sum_{l=1}^{M}\frac{2a_kb_l}{s_l-z_k}\\
			&+\frac1{2\pi}\int_{\H}\left(\frac{a_k}{x-z_k}+\frac{a_k}{\bar x-z_k}\right)\Lambda_{\delta,\eps}(x) e^{\Phi_{\delta,\eps}(x)}\d^2x+\frac1{\pi}\int_{\R}\frac{a_k}{t-z_k}\sigma_{\delta,\eps}(t) e^{\frac12\Phi_{\delta,\eps}(t)}\d t.
		\end{eqs}
		Likewise for a boundary insertion we have for $1\leq l\leq M+1$:
		\begin{eqs}\label{eq:L1_bord}
			&\Lc_{-1}^{(l)}[\Phireg]=\sum_{k=1}^{N}\frac{2b_la_k}{z_k-s_l}+\frac{2b_la_k}{\overline z_k-s_l}+\sum_{\substack{1\leq k\leq M+1\\ k\neq l}}\frac{2b_lb_k}{s_k-s_l}\\
			&+\frac1{2\pi}\int_{\H}\left(\frac{b_l}{x-s_l}+\frac{b_l}{\bar x-s_l}\right)\Lambda_{\delta,\eps}(x) e^{\Phi_{\delta,\eps}(x)}\d^2x+\frac1{\pi}\int_{\R}\frac{b_l}{t-s_l}\sigma_{\delta,\eps}(t) e^{\frac12\Phi_{\delta,\eps}(t)}\d t.
		\end{eqs}
	\end{lemma}
	\begin{proof}
		First note that the integrals on the right-hand side are absolutely convergent since $e^{\phireg}$ is integrable and that we integrate away from the singularities. Then using the explicit expression~\eqref{eq:green_H} of the Green's function together with Equation~\eqref{eq:key_formula2_reg} we get
		\begin{align*}
			\text{(LHS)-(RHS)}=2c_k\partial w_0(x_k)\left(2\sum_{l=1}^{N}a_l+\sum_{l=1}^Mb_l-2-\frac1{2\pi}\int_\Sigma\Lambda_{\delta,\eps}e^{\Phireg}\d^2x -\frac1{\pi}\int_{\partial\Sigma}\frac12\sigma_{\delta,\eps}e^{\frac12\Phireg}\d t \right)
		\end{align*}
		where (LHS) and (RHS) denote respectively the left-hand side and the right-hand side in Equation~\eqref{eq:L1}. We recognize the singular Euler characteristic for $\Div$ with $\chi(\D)=1$, $\eul=\sum_{l=1}^{N}a_l+\frac12\sum_{l=1}^Mb_l-1$. Since $\Lambda_{\delta,\eps}e^{\Phireg}\d^2x=R_{\gdiv}\d v_{\gdiv}$ and $\frac12\sigma_{\eps}e^{\frac12\Phireg}\d t=k_{\gdiv}\d l_{\gdiv}$ where $\gdiv$ solves Problem~\ref{prob:class} by the Gauss-Bonnet formula~\eqref{eq:gauss_bonnet_sing} (LHS)-(RHS) vanishes.
	\end{proof}
	In view of this statement we denote $\bm x=\left\{z_0,\cdots,z_N,\overline z_0,\cdots,\overline z_N,s_1,\cdots,s_{M+1}\right\}$ and likewise write $\bm c\coloneqq\{0,a_1,\cdots,a_N,0,a_1,\cdots,a_N,b_1\cdots,b_M,0\}$. We also set for $f$ smooth on $\Heps\cup\Reps$
	\begin{align*}
		\Itr\left[F(x)\right]\coloneqq \frac{1}{4\pi}\int_{\H} \left(f(x)+\overline{f(x)}\right))\Lambda_{\delta,\eps}(x) e^{\Phi_{\delta,\eps}(x)}\d^2x+\frac1{2\pi}\int_{\R}f(t)\sigma_{\delta,\eps}(t) e^{\frac12\Phi_{\delta,\eps}(t)}\d t.
	\end{align*}
	Equation~\eqref{eq:L1} then simplifies to, with the sum ranging from $0$ to $2N+M+2$,
	\begin{eqs}\label{eq:L1bis}
		&\Lc_{-1}^{(k)}[\Phireg]=\sum_{l\neq k}\frac{2c_kc_l}{x_l-x_k}+\Itr\left[\frac{2c_k}{x-x_k}\right].
	\end{eqs}
	
	\subsubsection{Descendants at level $2$}
	The derivative defined above corresponds to (the semi-classical limit) of a descendant of order $1$. We now turn to the descendants at level $2$, which allows to define the classical stress-energy tensor. This (regularized) descendant is defined as
	\begin{equation}
		\Lc_{-2}^{(k)}[\Phireg] \coloneqq \left(1-a_k\right)\partial^2\Phireg^{(k)}(z_k)-\frac12\left(\partial\Phireg^{(k)}(z_k)\right)^2
	\end{equation}
	for $x_k$ a bulk insertion, while for a boundary insertion $s_l$ we set
	\begin{equation}
		\Lc_{-2}^{(l)}[\Phireg] \coloneqq \frac12\left(1-b_l\right)\partial^2\Phireg^{(k)}(s_l)-\frac18\left(\partial\Phireg^{(l)}(s_l)\right)^2
	\end{equation}
	\begin{lemma}
		For any positive $\delta,\eps$ and $x_k$ any bulk or boundary insertion:
		\begin{eqs}\label{eq:L2}
			\Lc_{-2}^{(k)}[\Phireg]&=\sum_{l\neq k}\left(\frac{\delta_l}{(x_k-x_l)^2}+\frac{\Lc_{-1}^{(l)}[\Phireg]}{2(x_k-x_l)}\right)\\
			&+\Itr\left[\frac1{x_k-x}\left(\frac{c_k-1}{x_k-x}+\sum_{l\neq k}\frac{c_l}{x_l-x}\right)\right]-\frac12 \It^2\left[\frac1{(x-x_k)(y-x_k)}\right]
		\end{eqs}
		where $\delta_l\coloneqq -c_l\left(1+\frac12c_l\right)$ and $\It^2[f(x,y)]=\Itr[y\mapsto\Itr[f(x,y)]]$.
	\end{lemma}
	\begin{proof}
		Based on the previous computations we have for a bulk insertion
		\begin{align*}
			&\partial^2\Phireg^{(k)}(z_k)=\sum_{l\neq k}\frac{c_l}{(x_l-z_k)^2}+\It\left[\frac{1}{(x-z_k)^2}\right]\qt{and}\left(\partial\Phireg^{(k)}(z_k)\right)^2=\\
			&\sum_{l,m\neq k}\frac{c_lc_m}{(x_l-z_k)(x_m-z_k)}+2\Itr\left[\sum_{l\neq k}\frac{c_l}{(x-z_k)(x_l-z_k)}\right]+\It^2\left[\frac1{(x-z_k)(y-z_k)}\right].
		\end{align*}
		The claim then follows by distinguishing the cases $l=m$ and $l\neq m$ in the sum over $l,m\neq k$ and by writing $\frac{1}{(z-x_k)(x_l-x_k)}=\frac{1}{(z-x_l)(x_l-x_k)}+\frac{1}{(z-x_k)(x_l-z)}$ with $z$ either $x_m$ or $x$. For a boundary insertion similar computations remain valid.
	\end{proof}
	
	
	
	\subsection{Taking the limit: classical Ward identities}
	We now investigate the $\delta,\eps\to0$ limit of the expressions obtained above. The main issue lies in the fact that the singularities that appear there may not be integrable as $\eps,\delta\to0$ so that a proper limiting procedure is needed.
	\subsubsection{Taking the limit: first derivative}
	We first focus on the first derivative:
	\begin{lemma}\label{lemma:L1}
		For $0\leq k\leq N$ the limit $\Lc_{-1}^{(k)}[\Phidiv]\coloneqq\lim\limits_{\delta,\eps\to0}\Lc_{-1}^{(k)}[\Phireg]$ is well defined. 
		The same is valid for a boundary insertion by considering for $N+1\leq l\leq N+M+1$:
		\begin{equation}\label{eq:L1_lim}
			\Lc_{-1}^{(l)}[\Phidiv]\coloneqq\lim\limits_{\delta,\eps\to0}\left(\Lc_{-1}^{(l)}[\Phi_{\delta,\eps}]+\eval{x_l+\eps}{x_l-\eps}{-}{\frac{\sigma(t)}\pi e^{\frac12\Phireg(t)}}\right).
		\end{equation}
	\end{lemma}
	Here we have introduced the notation $\eval{b}{a}{\hspace{0.5cm}\pm}{F(t)}\coloneqq F(b)\pm F(a)$.
	\begin{proof}
		Let us consider a boundary insertion $x_l$. We write supp$(\sigma_{\eps})$ as the disjoint union of $\norm{t-x_l}>r$, $r\geq\norm{t-x_l}>\eps$ and $\eps\geq \norm{t}\geq \eps-\eps^{10}$. Over $r\geq\norm{t-x_l}>\eps$ we can write
		\[
		\frac{c_l}{t-x_l}\sigma(t) e^{\frac12\Phi_{\delta,\eps}(t)}=\partial_t \left(e^{c_l\ln\norm{t-x_l}}\right)\sigma(t)e^{\frac12\Phireg^{(l)}(t)},\qt{hence}
		\]
		\begin{align*}
			\int_{x_l-r}^{x_l-\eps}\frac{c_l}{t-x_l}\sigma(t) e^{\frac12\phi_{\delta,\eps}(t)}\d t=\eval{x_l-\eps}{x_l-r}{-}{\sigma(t) e^{\frac12\phi_{\delta,\eps}(t)}}-\int_{x_l-r}^{x_l-\eps}\partial_t \left(\sigma(t)e^{\frac12\phireg^{(l)}(t)}\right)e^{c_l\ln\norm{t-x_l}}\d t.		
		\end{align*}
		In the above expression all the terms except for $\sigma(x_l-\eps) e^{\frac12\phi_{\delta,\eps}(x_l-\eps)}$ admit a well-defined limit as $\eps\to0$. As a consequence we infer that the following limit exists:
		\[
		\lim\limits_{\delta,\eps\to0}\int_{x_l-r}^{x_l-\eps}\frac{c_l}{t-x_l}\sigma(t) e^{\frac12\phi_{\delta,\eps}(t)}\d t+\eval{x_l+\eps}{x_l-\eps}{-}{\sigma(t) e^{\frac12\phireg(t)}}.
		\]
		For the part over $\eps\geq \norm{t}\geq \eps-\eps^{10}$ we use boundedness of $\eps^{1-b_l}\frac1{t-x_l}\sigma(t) e^{\frac12\phireg(t)}$ to see that 
		\[
		\int_{\eps\geq \norm{t}\geq \eps-\eps^{10}}\frac1{t-x_l}\sigma(t) e^{\frac12\phireg(t)}\d t
		\]
		vanishes as $\eps\to0$. As for the integral over $\norm{t}>r$ it is uniformly bounded as $\eps\to0$ since $e^{\frac12\Phidiv}$ is integrable. This entails that following limit exists:
		\[
		\lim\limits_{\delta,\eps\to0}\int_{\R}\frac{c_l}{t-x_l}\sigma_\eps(t) e^{\frac12\phi_{\delta,\eps}(t)}\d t+\eval{x_l+\eps}{x_l-\eps}{-}{\sigma(t) e^{\frac12\phireg(t)}}\qt{ and is given by}
		\]
		\begin{align*}
			\int_{\norm{t-x_l}>r}\frac{c_l}{t-s_l}\sigma(t) e^{\frac12\phidiv(t)}\d t+\eval{x_l+r}{x_l-r}{-}{\sigma(t) e^{\frac12\phidiv(t)}}-\int_{\norm{t-x_l}<r}\partial_t \left(e^{\frac12\phidiv^{(l)}(t)}\right)\sigma(t)e^{c_l\ln\norm{t-x_l}}\d t.
		\end{align*}
		For the integral over $\H$ the same method remains valid by writing
		\[
		\left(\frac{c_l}{x-x_l}+\frac{c_l}{\overline x-x_l}\right)\Lambda_{\delta,\eps} e^{\phi_{\delta,\eps}(x)}=\left(\partial_x+\partial_{\bar x}\right) \left(e^{2c_l\ln\norm{x-x_l}}\right)\Lambda_{\delta,\eps} e^{\phireg^{(l)}(x)}.
		\]
		Then Stokes' formula over $B_{\delta}\coloneqq B(x_l,r)\cap\H_\delta$ on the above gives the boundary term
		\[
		\int_{\partial B(x_l,r)\cap\H_\delta}\Lambda e^{\phi_{\delta,\eps}(\xi)}\im\frac{d\bar\xi-d\xi}{2}
		\]
		because the contribution of $B(x_l,r)\cap\partial\H_\delta$ is zero since we considered the differential operator $\partial_x+\partial_{\bar x}$. This integral has a well-defined limit via Lemma~\ref{lemma:cont}. Hence the limit 
		in Equation~\eqref{eq:L1} exists. The same holds for a bulk insertion, with the boundary terms vanishing in the $\delta,\eps\to0$ limit, concluding the proof.
	\end{proof}	
	Recollecting terms we see that the descendant is explicitly given by
	\begin{eqs}\label{eq:L1_explicit}
		\Lc_{-1}^{(l)}[\Phidiv]=&\eval{x_l+r}{x_l-r}{-}{\frac{\sigma(t)}\pi e^{\frac12\Phidiv(t)}}+\frac1{4\pi}\int_{\partial B(x_l,r)\cap\H}\Lambda e^{\phidiv(\xi)}\im\left(d\bar\xi-d\xi\right)\\
		&+\Ir\left[\frac{2c_l}{x-x_l}\right]-\Is\left[\left(1+\mathds{1}_{x\in\H}\right)\partial_x\Phidiv^{(l)}(x)\right]
	\end{eqs}
	where in analogy with $\It$ we have set
	\begin{align*}
		&\Ir\left[f(x)\right]\coloneqq \frac1{4\pi}\int_{\norm{x-x_l}>r}\left(f(x)+\overline{f(x)}\right)\Lambda e^{\Phidiv(x)}\d^2x+\frac{1}{2\pi}\int_{\norm{t-x_l}>r}f(t)\sigma(t) e^{\frac12\Phidiv(t)}\d t\\
		&\Is\left[f(x)\right]\coloneqq \frac1{4\pi}\int_{\norm{x-x_l}<r}\left(f(x)+\overline{f(x)}\right)\Lambda e^{\Phidiv(x)}\d^2x+\frac{1}{2\pi}\int_{\norm{t-x_l}<r}f(t)\sigma(t) e^{\frac12\Phidiv(t)}\d t.
	\end{align*} 
	From Equation~\eqref{eq:L1_lim} it is independent of $r>0$ small enough. Moreover in the case where $c_l\geq 0$ we can take $r=0$ in the above which gives us the following explicit expression:
	\begin{equation}
		\Lc_{-1}^{(l)}[\Phidiv]=\frac{\sigma\left(x_l^+\right)-\sigma\left(x_l^-\right)}\pi e^{\frac12\Phidiv(x_l)}\mathds 1_{c_l=0}+\It\left[\frac{2c_l}{t-x_l}\right]\mathds 1_{c_l>0}.
	\end{equation}
	This descendant has the fundamental property that it coincides with the (weak) derivative of the Liouville action. Namely we will show in Proposition~\ref{prop:accessory} that in the weak sense of derivatives and for any $0\leq k\leq2N+M+2$
	\begin{equation}
		\Lc_{-1}^{(k)}[\Phidiv]=-\partial_{x_k}S_{\bm z,\bm a}.
	\end{equation}
	
	\subsubsection{Limit for the stress-energy tensor and Ward identities}
	We now turn to the analog limit for the descendant at level $2$. The method is the same as for the descendant $\Lc_{-1}$:
	\begin{lemma}\label{lemma:L2}
		For any $0\leq k\leq N$ the following limit exists and is well-defined:
		\begin{equation}
			\Lc_{-2}^{(k)}[\Phidiv]\coloneqq\lim\limits_{\delta,\eps\to0} \Lc_{-2}^{(k)}[\Phireg].
		\end{equation}
		For a boundary insertion the same statement holds true if we consider
		\begin{eqs}\label{eq:L2_lim}
			&\Lc_{-2}^{(l)}[\Phidiv]\coloneqq\lim\limits_{\delta,\eps,\rho\to0} \left(\Lc_{-2}^{(l)}[\Phireg]-\mathfrak{R}^{(l)}_{\delta,\eps}\right),\qt{where}\\
			&\mathfrak{R}^{(l)}_{\delta,\eps}=-\frac{1}{2\eps}\eval{x_l+\eps}{x_l-\eps}{+}{\frac{\sigma(t)}\pi e^{\frac12\Phireg(t)}}+\frac{\Lambda}{4\pi}\int_{(-\frac r\delta,\frac r\delta)}\frac{1}{1+t^2}  e^{\Phireg(s_l+\delta(t+i))}\d t.
		\end{eqs}
	\end{lemma}
	\begin{remark}
		This remainder can be simplified to
		\begin{eqs}\label{eq:L2_expr}
			&-\frac{1}{2\eps}\eval{x_l+\eps}{x_l-\eps}{+}{\frac{\sigma(t)}\pi e^{\frac12\Phireg(t)}}+2^{c_l}\sin\left(\pi c_l\right)\frac{\Gamma(c_l)\Gamma(1-2c_l)}{\Gamma(1-c_l)}\mathds{1}_{c_l\leq 0}\frac{\Lambda}{4\pi}e^{\Phireg(x_l+i\delta)}+l.o.t.
		\end{eqs}
		by using that the following integral can be evaluated (see~\cite[Lemma A.1]{Cer_HEM}):
		\[
		\int_\R \left(1+t^2\right)^{u-1}\d t=2^{u}\sin\left(\pi u\right)\frac{\Gamma(u)\Gamma(1-2u)}{\Gamma(1-u)}\cdot
		\]
	\end{remark}
	This statement is actually a consequence of the (proof of the) following Ward identities:
	\begin{proposition}\label{prop:L2}
		For a bulk or a boundary insertion the local Ward identities hold:
		\begin{equation}\label{eq:ward}
			\Lc_{-2}^{(k)}[\Phidiv]=\sum_{l\neq k}\left(\frac{\delta_l}{(x_k-x_l)^2}+\frac{\Lc_{-1}^{(l)}[\Phidiv]}{2(x_k-x_l)}\right).
		\end{equation}
	\end{proposition}
	\begin{proof}
		We rely on Equation~\eqref{eq:L2}. First of all for $\eps-\eps^{10}<\norm{x-x_k}<\eps$ the integrals
		\[
		\int_{x_k+\eps-\eps^{10}}^{x_k+\eps}\frac1{x_k-x}\left(\frac{c_k-1}{x_k-x}+\sum_{l\neq k}\frac{c_l}{x_l-x}\right)e^{\frac12\Phireg(x)}\d x
		\]
		are $o(1)$. Hence we can replace in Equation~\eqref{eq:L2} $\Lambda_{\delta,\eps}$ and $\sigma_\eps$ by respectively $\Lambda\mathds1_{\Heps}$ and $\sigma\mathds1_{\Reps}$.
		Then we use the following consequence of Equation~\eqref{eq:L1}: 
		\[
		\partial_x\left(\frac1{x_k-x}e^{\Phireg(x)}\right)=\frac1{x_k-x}\left(\frac{c_k-1}{x_k-x}+\sum_{l\neq k}\frac{c_l}{x_l-x}\right)e^{\Phireg(x)}+\It\left[\frac{1}{(x_k-x)(y-x)}\right].	
		\]
		As a consequence, using Equation~\eqref{eq:L2}, we have that the difference between the left- and right-hand sides in Equation~\eqref{eq:ward} is given by
		\begin{align*}
			&\Lc_{-2}^{(k)}[\Phireg]-\sum_{l\neq k}\left(\frac{\delta_l}{(x_k-x_l)^2}+\frac{\Lc_{-1}^{(l)}[\Phireg]}{2(x_k-x_l)}\right)=\\
			&\It\left[\partial_x\left(\frac1{x_k-x}e^{\Phireg(x)}\right)e^{-\Phireg(x)}\right]-\It^2\left[\frac{1}{(x_k-x)(y-x)}+\frac{1}{2(x_k-x)(x_k-y)}\right].
		\end{align*} 
		Since $\frac{1}{(x_k-x)(y-x)}+\frac{1}{(x_k-y)(x-y)}=-\frac{1}{(x_k-x)(x_k-y)}$, by symmetry in $x,y$ the two-fold integral vanishes in the $\delta,\eps\to0$ limit. Therefore it only remains to understand the remaining term $\It\left[\partial_x\left(\frac1{x_k-x}e^{\Phireg(x)}\right)e^{-\Phireg(x)}\right]$. By integration by parts/Stokes' formula it is equal to
		\begin{align*}
			&\frac \im{8\pi}\int_{\R}\left(\frac{1}{t+\im\delta-x_k}-\frac{1}{t-\im\delta-x_k}\right)\Lambda e^{\Phireg(t+\im\delta)}\d t+\sum_{l=1}^{N}\frac \im{8\pi}\oint_{\partial B(z_l,\eps)}\left(\frac{\d\bar \xi}{\xi-x_k}+\frac{\d\xi}{\overline \xi-x_k}\right)\Lambda e^{\Phireg(\xi)}\\
			&+\sum_{l=1}^M\eval{s_l+\eps}{s_l-\eps}{-}{\frac{1}{t-x_k}\frac{\sigma(t)}{2\pi}e^{\frac12\Phireg(t)}}.
		\end{align*}
		If $x_k\in\R$, in the integral over $\R$ we can make the change of variable $t\to \frac1\delta(t+x_k)$ to get
		\[
		\frac1{4\pi}\int_\R\frac{1}{1+t^2}\Lambda_{\delta,\eps}(x_k+\delta(t+\im))e^{\Phireg(x_k+\delta(t+\im))}\d t.
		\]
		If $x_k\in\H$ then the integral over $\R$ is a $o(1)$ since $e^{\Phireg}$ is integrable uniformly in $\delta,\eps$.
		For the integrals over the circles $\partial B(z_l,\eps)$, if $x_k$ is a bulk insertion for $z_l=x_k$ we expand $e^{\Phireg(\xi)}=\norm{\xi-z_l}^{2a_l}\left(\lambda_{\delta,\eps}+\lambda'_{\delta,\eps}(\xi+\overline \xi-z_l-\overline{z}_l)+\mc O(\norm{x-z_l}^2)\right)$ where $\lambda_{\delta,\eps}$, $\lambda'_{\delta,\eps}$ and the $\mc O$ are uniformly bounded in $\delta,\eps$. Since $\oint_{\partial B(x_k=z_l,\eps)}\norm{\xi-x_k}^{2a_k}(\xi-x_k)^p(\overline \xi-\overline x_k)^q\frac{\d\bar \xi}{\xi-x_k}=0$ for $p-q\neq 2$ the integral over $\partial B(z_l,\eps)$ is a $\mc O(\eps^{2a_l+2})$ hence vanishes in the limit. The same applies if $z_l\neq x_k$ so the sum for $1\leq l\leq N$ is a $o(1)$. Hence recollecting terms:
		\begin{align*}
			&\Lc_{-2}^{(k)}[\Phireg]-\sum_{l\neq k}\left(\frac{\delta_l}{(x_k-x_l)^2}+\frac{\Lc_{-1}^{(l)}[\Phireg]}{2(x_k-x_l)}\right)=o(1)+\sum_{\substack{1\leq l\leq M\\ s_l\neq x_k}}\eval{s_l+\eps}{s_l-\eps}{-}{\frac{1}{t-x_k}\frac{\sigma(t)}{2\pi}e^{\frac12\Phireg(t)}}\\
			&+\mathds1_{x_k\in\R}\left(\frac{1}{4\pi}\int_\R\frac{1}{1+t^2}\Lambda_{\delta,\eps}(x_k+\delta(t+\im))e^{\Phireg(x_k+\delta(t+\im))}\d t+\eval{x_k+\eps}{x_k-\eps}{-}{\frac{1}{t-x_k}\frac{\sigma(t)}{2\pi}e^{\frac12\Phireg(t)}}\right).
		\end{align*} 
		The sum ranging over $1\leq l\leq M$ corresponds to the remainder terms that appear in the definition of $\mc L_{-1}^{(l)}$ and featured in the RHS of Equation~\eqref{eq:L2}. As for the second line, we see that the $\R$-integral scales like the one of the remainder term in Lemma~\ref{lemma:L2} while the other quantity corresponds to the first term in the remainder of Equation~\eqref{eq:L2_lim}. In brief
		\begin{align*}
			&\Lc_{-2}^{(k)}[\Phireg]-\mathfrak{R}_{\delta,\eps}^{(k)}=\sum_{l\neq k}\left(\frac{\delta_l}{(x_k-x_l)^2}+\frac{\Lc_{-1}^{(l)}[\Phireg]}{2(x_k-x_l)}\right)+\sum_{\substack{1\leq l\leq M\\ s_l\neq x_k}}\frac{1}{s_l-x_k}\eval{s_l+\eps}{s_l-\eps}{-}{\frac{\sigma(t)}{2\pi}e^{\frac12\Phireg(t)}}+o(1).
		\end{align*} 
		Recollecting terms we see that the limit in Equation~\eqref{eq:L2} does indeed exist and that it is given by the Ward identity from Equation~\eqref{eq:ward} with the $\Lc_{-1}^{(l)}$'s defined by Lemma~\ref{lemma:L1}.
	\end{proof}
	
	Thanks to the previous statements we can define the stress-energy tensor $T$. To this end we set for $x=z_0\in\H\setminus\bm z$ in the bulk and for $t=s_{M+1}\in\R\setminus\bm z$ on the boundary
	\begin{eqs}
		&T(x)\coloneqq \partial^2\Phidiv(x)-\frac12\left(\partial\Phidiv(x)\right)^2,\\
		&T(t)\coloneqq -\frac{\Lambda}4 e^{\Phidiv(t)}+\lim\limits_{\delta,\eps\to0}\left(\frac12\partial^2\Phireg(t)-\frac18\left(\partial\Phireg(t)\right)^2+\frac{1}{\eps}\frac{\sigma(t)}{2\pi}e^{\frac12\Phireg(t)}\right).
	\end{eqs}
	Then in agreement with Lemma~\ref{lemma:L2} and Proposition~\ref{prop:L2} we have for $z\in\overline{\H}$
	\begin{equation}\label{eq:weight_accessory}
		T(z)=\sum_{k=1}^{2N+M}\left(\frac{\delta_k}{(z-x_k)^2}+\frac{\bm c_k}{z-x_k}\right)\text{ with }\delta_k=-c_k(1+\frac12c_k)\text{ and }\bm c_k=\frac12 \Lc_{-1}^{(k)}[\Phidiv].
	\end{equation}
	
	
	
	\section{The semi-classical limit of boundary Liouville theory}\label{sec:semi_class}
	The geometry that we have described in the previous section was completely deterministic, \textit{i.e.} the data of a Riemannian surface $(\Sigma,g)$, a divisor $\Div$ and non-positive scalar and geodesic curvatures defines a unique metric on $\Sigma$ representing $\Div$ and with prescribed curvatures. This metric can be described as the unique minimum of the Liouville action. 
	
	Based on this action Liouville Conformal Field Theory (CFT hereafter) provides a way of defining a \textit{random} geometry on $\Sigma$. The level of randomness is given by the coupling constant $\gamma\in(0,2)$, the classical theory corresponding to taking the \textit{semi-classical limit} $\gamma\to0$. 
	
	In this section we explain how to make this assertion rigorous (Theorem~\ref{thm:semi_classical}). To do so we first recall the probabilistic definition of (boundary) Liouville CFT and provide some key estimates for the so-called \textit{derivatives Gaussian Multiplicative Chaos measure}.
	
	
	
	\subsection{Definition of boundary Liouville theory}
	The definition of (boundary) Liouville CFT relies on a probabilistic framework. We recall here some of the material from~\cite{DKRV,DRV16,GRV16,HRV16,Wu} to which we refer for additional details. Hereafter we let $\Sigma$ be a compact connected Riemannian surface equipped with a smooth Riemannian metric $g$.
	
	\subsubsection{Gaussian Free Field and Gaussian Multiplicative Chaos}
	The Gaussian Free Field (GFF hereafter) is the random element $\X$ of $H^{-1}(\Sigma,g)$~\cite{Ber, PW} with covariance kernel
	\begin{equation}
		\expect{\X(x)\X(y)}=G_g(x,y)\qt{for all $x\neq y$ in $\overline{\Sigma}$}
	\end{equation}
	where $G_g$ is the Green's function from Subsection~\ref{subsec:green}. In particular if $\partial\Sigma$ is non-empty the GFF has Neumann boundary conditions. Because $\X$ is not a true function we regularize it by setting for $\rho>0$  $\X_\rho$ to be the average of $\X$ over geodesic (semi)-circles of radii $\rho$.
	
	The theory of Gaussian Multiplicative Chaos~\cite{Kah,DS10, RV08, Ber, Sha} (GMC in the sequel) provides a meaning to the exponential of the GFF $\X$. Namely for $\gamma\in(0,2)$ the following limits hold in probability and in the sense of weak convergence of Radon measures on $\Sigma$:
	\begin{equation}
		e^{\gamma\X}\d v_g\coloneqq \lim\limits_{\rho\to0}\rho^{\frac{\gamma^2}{2}}e^{\gamma\X_\rho}\d v_g\qt{and}e^{\frac\gamma2\X}\d v_g\coloneqq \lim\limits_{\rho\to0}\rho^{\frac{\gamma^2}{4}}e^{\frac\gamma2\X_\rho}\d v_g.
	\end{equation}
	Thanks to Lemma~\ref{lemma:approx_green} we know that for $g'=e^{\varphi}g$ in the conformal class of $g$
	\begin{equation}
		e^{\gamma\X}\d v_{g'}\eqlaw e^{\gamma\left(\X+\frac Q2\ln\varphi\right)}\d v_{g}\qt{and}e^{\frac\gamma2\X}\d v_{g'}\eqlaw e^{\frac\gamma2\left(\X+\frac Q2\ln\varphi\right)}\d v_{g}
	\end{equation}
	with $Q=\frac\gamma2+\frac2\gamma$. This property is crucial to get the Weyl anomaly of correlation functions, which allows in particular to assume without any loss of generality that $g$ is a uniform metric of type $I$. We will work under this assumption in the sequel. 
	
	\subsubsection{Liouville correlation functions}
	Let us now pick a divisor $\Diva$ of the form
	\begin{eqs}\label{eq:seiberg1}
		&\Diva=\sum_{k=1}^N\alpha_k z_k+\frac12\sum_{l=1}^M\beta_l s_l\qt{with $z_k\in\Sigma$ and $s_l\in\partial\Sigma$ for all $k,l$ and where}\\
		&\alpha_k<Q\qt{and}\beta_l<Q\qt{for all $1\leq k\leq N$ and $1\leq l\leq M$.}
	\end{eqs}
	In analogy with the singular Euler characteristic we set for such a divisor and $\gamma\in(0,2)$
	\begin{equation}
		\chi_\gamma(\Diva)\coloneqq \frac 12\left(Q\chi(\Sigma)-\sum_{k=1}^N\alpha_k-\frac12\sum_{l=1}^M\beta_l\right).
	\end{equation}
	Hereafter we assume that $\chi_\gamma(\Diva)<0$; together with Equation~\eqref{eq:seiberg1} these are the Seiberg bounds under which Liouville correlation functions are well-defined. Like before we write $\bm z=(z_1,\cdots,z_{N+M})\coloneqq(z_1,\cdots,z_N,s_1,\cdots,s_M)$ and $\bm\alpha=(\alpha_1,\cdots,\alpha_{N+M})\coloneqq (\alpha_1,\cdots,\alpha_N,\beta_1,\cdots,\beta_M)$.
	
	We also choose $\mu>0$ continuous over $\Sigma$ as well as $\mu_\partial:\partial\Sigma\to\R^+$ continuous on each connected component of $\partial\Sigma\setminus\bm z$. Liouville correlation functions are then~\cite{DKRV,DRV16,GRV16,HRV16,Wu} 
	\begin{equation}
		\ps{\prod_{k=1}^NV_{\alpha_k}(z_k)\prod_{l=1}^MV_{\beta_l}(s_l)}_{\gamma,\mu,\mu_\partial}\coloneqq \lim\limits_{\rho\to0}\frac1{\mc Z_g}\int_\R\expect{e^{-\tilde S^{(\rho)}_{\bm z,\bm\alpha}(\X_\rho+c)}}\d c,\qt{with}
	\end{equation}
	\begin{eqs}
		\tilde S^{(\rho)}_{\bm z,\bm\alpha}(\phi)\coloneqq &\frac1{4\pi}\int_\Sigma\left(Q R_g\phi+2\mu \rho^{\frac{\gamma^2}{2}}e^{\gamma\phi}\right)\d v_g+\frac1{2\pi}\int_{\partial\Sigma}\left(Q k_g\phi+4\mu_\partial \rho^{\frac{\gamma^2}{4}}e^{\frac\gamma2\phi}\right)\d l_g\\
		&-\sum_{k=1}^N\alpha_k\phi(z_k)-\frac12\sum_{l=1}^M\beta_l\phi(s_l)-\left(\frac12\sum_{k=1}^N\alpha_k^2+\frac14\sum_{l=1}^M\beta_l^2\right)\ln\rho
	\end{eqs}
	and $\mc Z_g\coloneqq\left(\frac{\det'\left(-\frac1{2\pi}\Delta_g\right)}{v_g(\Sigma)}\right)^{\frac12}$ a regularized determinant.  
	More generally, we can define the law of the Liouville field $\Phi$ to be the Borel measure on $H^{-1}(\Sigma,g)$ uniquely defined (via Riesz-Markov-Kakutani representation theorem) by the assignment for $F$ continuous and positive over $H^{-1}(\Sigma,g)$ of
	\begin{equation}
		\ps{F(\Phi)\prod_{k=1}^NV_{\alpha_k}(z_k)\prod_{l=1}^MV_{\beta_l}(s_l)}_{\gamma,\mu,\mu_\partial}\coloneqq \lim\limits_{\rho\to0}\frac1{\mc Z_g}\int_\R\expect{F(\X_\rho+c)e^{-\tilde S^{(\rho)}_{\bm z,\bm\alpha}(\X_\rho+c)}}\d c.
	\end{equation}
	Hereafter we will denote by $\d\mathbb P_{\gamma,\bm z,\bm a}$ the law of the Liouville field thus defined, that is
	\begin{equation}
		\mathbb E_{\gamma,\bm z,\bm a}\left[F(\Phi)\right]\coloneqq \frac{\ps{F(\Phi)\prod_{k=1}^NV_{\alpha_k}(z_k)\prod_{l=1}^MV_{\beta_l}(s_l)}_{\gamma,\mu,\mu_\partial}}{\ps{\prod_{k=1}^NV_{\alpha_k}(z_k)\prod_{l=1}^MV_{\beta_l}(s_l)}_{\gamma,\mu,\mu_\partial}}\cdot
	\end{equation}
	
	Using Girsanov's theorem 
	the correlation functions and the law of the Liouville field can be written in a more convenient way. To be more explicit let $W$ be as in Lemma~\ref{lemma:approx_green} and set
	\begin{equation}
		\Hdiva\coloneqq \sum_{k=1}^{N+M}\alpha_k G_g(\cdot,z_k)\qt{and}G(\bm z,\bm\alpha)\coloneqq -\frac12\sum_{k\neq l}\alpha_l\alpha_lG_g(z_k,z_l)-\frac12\sum_{k=1}^{N+M}\alpha_k^2 W(z_k).
	\end{equation}
	Then if $g$ is uniform of type $I$ the correlation functions can be written as
	\begin{eqs}
		&\ps{\prod_{k=1}^NV_{\alpha_k}(z_k)\prod_{l=1}^MV_{\beta_l}(s_l)}_{\gamma,\mu,\mu_\partial}=e^{-G(\bm z,\bm\alpha)}\frac1{\mc Z_g}\int_\R e^{-\frac2\gamma\chi_\gamma c}\E\Bigg[F(\X+\Hdiva+c)\\
		&\hspace{2.5cm}\exp\left(-\frac1{2\pi}\int_{\Sigma}\mu e^{\gamma\left(\X+\Hdiva+c\right)}\d v_g-\frac2{\pi}\int_{\partial\Sigma}\mu_\partial e^{\frac\gamma2\left(\X+\Hdiva+c\right)}\d l_g\right)\Bigg]\d c.
	\end{eqs}
	
	
	
	\subsection{Additional probabilistic tools} 
	The semi-classical limit of Liouville CFT relies on probabilistic techniques related to the framework introduced above. We provide here the main tools used in this respect: derivative(s) Gaussian Multiplicative Chaos measures. Part of the material presented here is adapted from~\cite{LRV_semi1, LRV_semi2} to which we refer for more details.
	
	\subsubsection{Derivative(s) Gaussian Multiplicative Chaos}
	The derivative GMC measure was introduced in~\cite{LRV_Mabuchi}. It (formally) corresponds to taking the $\gamma$ derivative of the GMC measures $e^{\gamma\X}\d v_g$ and $e^{\frac\gamma2\X}\d l_g$. We propose here to generalize this definition to higher derivatives; as we will see second order ones naturally appear in the $\gamma\to0$ limit of Liouville CFT.
	
	Let $k$ be a non-negative integer and $\gamma\in[0,1)$. We define for any $\rho>0$ the signed measures
	\begin{equation}
		\mc D_{\gamma,\rho}^{(k)}[\X]\d v_g\coloneqq \left(\frac{\partial}{\partial\gamma}\right)^k\left( \rho^{\frac{\gamma^2}2}e^{\gamma \X_\rho}\right)\d v_g,\quad \mc D_{\partial,\gamma,\rho}^{(k)}[\X]\d l_g\coloneqq \left(\frac{\partial}{\partial\gamma}\right)^k\left( \rho^{\frac{\gamma^2}2}e^{\gamma \X_\rho}\right)\d l_g.
	\end{equation}
	\begin{proposition}\label{prop:DGMC}
		Let $\gamma\in[0,1)$. Then for any non-negative integer $k$ and $f$ continuous and bounded over $\overline\Sigma$ the following limits exist in $L^2$:
		\begin{equation}
			\int_\Sigma f\mc D_{\gamma}^{(k)}[\X]\d v_g\coloneqq\lim\limits_{\rho\to0}\int_\Sigma f\mc D_{\gamma,\rho}^{(k)}[\X]\d v_g,\quad\int_\Sigma f\mc D_{\partial,\gamma}^{(k)}[\X]\d l_g\coloneqq\lim\limits_{\rho\to0}\int_\Sigma f\mc D_{\partial,\gamma,\rho}^{(k)}[\X]\d l_g.
		\end{equation}
	\end{proposition}
	\begin{proof}
		For the first point let us take $f$ continuous and bounded over $\overline\Sigma$. Then for $\rho,\rho'>0$
		\begin{align*}
			&\expect{\left(\int_\Sigma f\mc D_{\gamma,\rho}^{(n)}\d v_g-\int_\Sigma f\mc D_{\gamma,\rho'}^{(n)}\d v_g\right)^2}=\partial_{\gamma_1}^n\partial_{\gamma_2^n
			}\\
			&\int_{\Sigma^2}f(x)f(y)\expect{\left(\rho^{\frac{\gamma_1^2}2}e^{\gamma_1\X_\rho(x)}-\rho'^{\frac{\gamma_1^2}2}e^{\gamma_1\X_{\rho'}(x)}\right)\left(\rho^{\frac{\gamma_2^2}2}e^{\gamma_2\X_\rho(x)}-\rho'^{\frac{\gamma_2^2}2}e^{\gamma_2\X_{\rho'}(x)}\right)}\d v_g(x)\d v_g(y)
		\end{align*}
		evaluated at $\gamma_1=\gamma_2=\gamma$. The term on the second line is equal to
		\begin{align*}
			&\int_{\Sigma^2}f(x)f(y)e^{\frac{\gamma_1^2}{2}W_\rho(x)}\left(e^{\frac{\gamma_2^2}{2}W_\rho(y)}e^{\gamma_1\gamma_2G_{\rho,\rho}(x,y)}-e^{\frac{\gamma_2^2}{2}W_{\rho'}(y)}e^{\gamma_1\gamma_2G_{\rho,\rho'}(x,y)}\right)\d v_g(x)\d v_g(y)\\
			&-\int_{\Sigma^2}f(x)f(y)e^{\frac{\gamma_1^2}{2}W_{\rho'}(x)}\left(e^{\frac{\gamma_2^2}{2}W_\rho(y)}e^{\gamma_1\gamma_2G_{\rho',\rho}(x,y)}-e^{\frac{\gamma_2^2}{2}W_{\rho'}(y)}e^{\gamma_1\gamma_2G_{\rho',\rho'}(x,y)}\right)\d v_g(x)\d v_g(y)
		\end{align*}
		where $G_{\rho,\rho'}(x,y)\coloneqq \left(G_\rho(x,\cdot)\right)_{\rho'}(y)$ and $W_\rho(x)\coloneqq G_\rho(x,x)+\ln\rho=W(x)+o(1)$ (see Lemma~\ref{lemma:approx_green}). Now for any positive integer $k$ and for $\gamma<1$ the following integral
		\[
		\int_\Sigma \norm{G_g(x,y)}^{k}e^{\gamma^2G_g(x,y)}\d v_g(x)\d v_g(y)
		\]
		is absolutely convergent (we only need $\gamma<\sqrt2$ in a compact inside the bulk, but close to the boundary we must have $\gamma<1$). This implies that the sequence of $\left(\int_\Sigma f\mc D_{\gamma,\rho}^{(n)}\d v_g\right)_{\rho>0}$ is Cauchy in $L^2$ as $\rho\to0$. The same computations remain valid for the boundary measure as soon as $\gamma<\sqrt2$, showing that for any $f$ continuous and bounded over $\overline\Sigma$ the sequence \\
		$\left(\int_\Sigma f\mc D_{\gamma,\rho}^{(n)}\d v_g+\int_{\partial\Sigma} f\mc D_{\partial\gamma,\rho}^{(n)}\d l_g\right)_{\rho>0}$ is Cauchy in $L^2$, hence convergent as $\rho\to0$.
		
		The second statement follows from the first one by definition.
	\end{proof}
	
	In analogy with~\cite{LRV_Mabuchi} we call such measures \textit{Derivatives Gaussian Multiplicative Chaos measures} (DGMC). They admit the following alternative representation: for $\rho>0$ set
	\begin{equation}
		e^{\gamma\X_\rho-\frac{\gamma^2}{2}\ln\rho}=\sum_{n\geq0}\frac{\gamma^n}{n!}H_n^{(\rho)}(\X_\rho)\qt{and}e^{\frac\gamma2\X_\rho-\frac{\gamma^2}{4}\ln\rho}=\sum_{n\geq0}\frac{(\gamma/2)^n}{n!}H_{\partial,n}^{(\rho)}(\X_\rho)
	\end{equation}
	for some (generalized) Hermite polynomials $H_n^{(\rho)}$ (explicitly $H_n^{(\rho)}=\ln(\rho)^{\frac n2}\Herm_n\left(\frac{\X_\rho}{\sqrt{\ln\rho}}\right)$ where the $\Herm_n$'s are the usual Hermite polynomials). More generally around $\gamma_0\in(0,1)$ we write 
	\begin{equation}
		e^{\gamma\X_\rho-\frac{\gamma^2}{2}\ln\rho}=\sum_{n\geq0}\frac{(\gamma-\gamma_0)^n}{n!}H_{n,\gamma_0}^{(\rho)}(\X_\rho)e^{\gamma_0\X_\rho-\frac{\gamma_0^2}{2}\ln\rho}
	\end{equation}
	and likewise for the boundary measure.
	Then Proposition~\ref{prop:DGMC} is equivalent to the fact that for any non-negative $n$ and $\gamma\in(0,1)$ the following limits exist:
	\begin{eqs}
		&H_{n,\gamma}(\X)e^{\gamma\X}\d v_g\coloneqq \lim\limits_{\rho\to0}H_{n,\gamma}^{(\rho)}(\X_\rho)e^{\gamma\X_\rho-\frac{\gamma^2}{2}\ln\rho}\d v_g\qt{and}\\
		&H_{\partial,n,\gamma}(\X)e^{\frac\gamma2\X}\d l_g\coloneqq\lim\limits_{\rho\to0}H_{\partial,n,\gamma}^{(\rho)}(\X_\rho)e^{\frac\gamma2\X_\rho-\frac{\gamma^2}{4}\ln\rho}\d l_g.
	\end{eqs}
	Besides for any non-negative $n$ and $\gamma\in[0,1)$, $\mc D_\gamma^{(n)}[\X]\d v_g=H_{n,\gamma}(\X)e^{\gamma\X}\d v_g$ and for $n=1,2$:
	\[
	H_{1,\gamma}^{(\rho)}(\X_\rho)=\X_\rho+\gamma\ln\rho\qt{and}H_{2,\gamma}^{(\rho)}(\X_\rho)=\X_\rho^2+\ln\rho+2\ln\rho\left(\gamma\X_\rho+\frac{\gamma^2}2\ln\rho\right).
	\]
	\subsubsection{Some properties of the DGMC measures}	
	Having defined them, we record some elementary facts about DGMC measures. First of all we have the following Taylor's formula:
	\begin{lemma}\label{lemma:DGMC_DL}
		For any $\gamma \in (0,1)$, $\gamma_0\in[0,1)$, $n\geq 0$ and $f$ continuous and bounded over $\overline\Sigma$
		\begin{equation}
			\int_\Sigma fe^{\gamma\X}\d v_g=\sum_{k=0}^{n-1}\frac{(\gamma-\gamma_0)^k}{k!}\int_\Sigma f\mc D_{n,\gamma_0}\d v_g+\int_{\gamma_0}^{\gamma}\frac{(\gamma-\gamma')^{n-1}}{(n-1)!}\left(\int_\Sigma f \mc D_{\gamma'}^{(n)}\d v_g\right)\d\gamma'.
		\end{equation}
	\end{lemma}
	\begin{proof}
		For any positive $\rho$, this equality is simply Taylor's formula. Taking the $\rho\to 0$ limit is granted by Proposition~\ref{prop:DGMC} thanks to which we have $L^2$ convergence of both sides.
	\end{proof}
	This property is elementary but nonetheless crucial to take the semi-classical limit $\gamma\to0$.
	Another key property in this perspective is the existence of uniform bounds in $\gamma\in(0,1)$ for functionals of DGMC measures. Assume that $\Lambda\geq 0$ and $\sigma\geq0$ are bounded respectively over $\Sigma$ and $\partial\Sigma$, with $\Lambda\not\equiv0$. Let $\Div$ be a divisor like in the previous section and set
	\begin{equation}
		\ps{f}_{\bm z,\bm a}\coloneqq \frac1{4\pi}\int_\Sigma f\Lambda\d v_{\gdiv}+\frac{1}{4\pi}\int_{\partial\Sigma} f\sigma\d l_{\gdiv}\qt{and}m_{\bm z,\bm a}(f)\coloneqq \frac{\ps{f}_{\bm z,\bm a}}{\ps 1_{\bm z,\bm a}}
	\end{equation}
	for $f$ continuous and bounded over $\overline{\Sigma}$ and with $\gdiv=e^{\phidiv}g$ the solution of Problem~\ref{prob:class} given by Theorem~\ref{thm:prescribe}.
	In addition to ensure that the semi-classical limit is well-defined we will need existence of some negative exponential moments of these DGMC measures: 
	\begin{lemma}\label{lemma:lim_der_GMC}
		For any $c\in\R$ we have the uniform bound
		\begin{equation}\label{eq:DGMC_sup}
			\sup\limits_{\gamma\in(0,\frac12)}\expect{\exp\left(-\frac1{4\pi}\int_\Sigma \Lambda \tilde{\mc D}^{(2)}_\gamma[\X+c]\d v_{\gdiv}- \frac1{4\pi}\int_{\partial\Sigma} \sigma\tilde{\mc D}^{(2)}_\gamma[\X+c]\d l_{\gdiv}\right)}<\infty\qt{where}
		\end{equation}
		\[
		\tilde{\mc D}^{(2)}_\gamma[\X+c]\d v_{\gdiv}\coloneqq 2\int_0^{1}\left(\mc  D_{u\gamma}^{(2)}[\X+c]\d v_{\gdiv}\right)u\d u,\text{ }\tilde{\mc D}^{(2)}_\gamma[\X+c]\d l_{\gdiv}\coloneqq 2\int_0^{1}\left(\mc  D^{(2)}_{\partial,u\gamma}[\X+c]\d l_{\gdiv}\right)u\d u.
		\]
		Moreover for any $F$ continuous and bounded over $H^{-1}(\Sigma,g)$
		\begin{eqs}\label{eq:gam_DGMC}
			\lim\limits_{\gamma\to0}&\expect{F(\X)\exp\left(-\frac1{4\pi}\int_\Sigma \Lambda \tilde{\mc D}^{(2)}_\gamma[\X+c]\d v_{\gdiv}-\frac1{4\pi}\int_{\partial\Sigma} \sigma\tilde{\mc D}^{(2)}_\gamma[\X+c]\d l_{\gdiv}\right)}\\
			=&\expect{F(\X)\exp\left(-\ps{H_{2}[\X+c]}_{\bm z,\bm a}\right)}.
		\end{eqs}
	\end{lemma}
    \begin{proof}
        Since the arguments developed are very similar to that of~\cite[Lemma 3.6]{LRV_semi1} (see also~\cite{Hua_Phi} and the proof of~\cite[Proposition 4.2]{LRV_semi2}), we will be brief and highlight the main differences with our setting.
	To start with for any $\gamma\in(0,1)$ the expectation term in Equation~\eqref{eq:DGMC_sup} is well-defined by using Lemma~\ref{lemma:DGMC_DL} together with Girsanov's theorem (see the proof of Theorem~\ref{thm:semi_classical}). Moreover since we have the almost sure convergence $\tilde{\mc D}^{(2)}_\gamma[\X+c]\d v_{\gdiv}\to H_2[\X+c]\d v_{\gdiv}$ as $\gamma\to0$ via Lemma~\ref{lemma:DGMC_DL} (and jointly with the boundary measure) we only need to prove that the expectation term in Equation~\eqref{eq:DGMC_sup} is uniformly bounded for $\gamma$ close to $0$. To this end we set 
	\[
		R_{\gamma,\rho}\coloneqq \int_\Sigma \Lambda \tilde{\mc D}^{(2)}_{\gamma,\rho}[\X]\d v_{\gdiv}+\int_{\partial\Sigma} \sigma\tilde{\mc D}^{(2)}_{\gamma,\rho}[\X]\d l_{\gdiv}\qt{and}\rho_\gamma\coloneqq e^{-\gamma^{-\frac18}}
	\]
	so that $\expect{\X_{\rho_\gamma}(x)^2}=\gamma^{-\frac18}+W_{\rho_\gamma}(x)+o(1)$. Then in the same fashion as in~\cite[Equation (3.25)]{LRV_semi1} or~\cite[Equation (5.14)]{LRV_semi2} we have the uniform bound $\sup\limits_{\gamma\in(0,\frac12)}\expect{e^{-R_{\gamma,\rho_\gamma}}}<\infty$. Indeed
	\begin{align*}
		\expect{e^{-R_{\gamma,\rho_\gamma}}}=\expect{e^{-R_{\gamma,\rho_\gamma}}\left(\mathds 1_{\mc A}+\mathds 1_{\mc A^c}\right)}\qt{where}\mc A\coloneqq \left\{\ps{\X_\rho(\cdot)^2\mathds1_{\X_\rho(\cdot)^2>\gamma^{-\frac12}}}_{\bm z,\bm a}\geq e^{-\gamma^{-\frac18}}\right\}.
	\end{align*}
	Then by Markov inequality $\mathbb P(\mc A)\leq e^{\gamma^{-\frac18}} \expect{\ps{\X_\rho(\cdot)^2\mathds1_{\X_\rho(\cdot)^2>\gamma^{-\frac12}}}_{\bm z,\bm a}}\leq c e^{\gamma^{-\frac18}-c\gamma^{-\frac38}}$ for some $c>0$ where the last bound is obtained by Gaussian computations. Likewise under the event $\mc A^c$ we can bound $R_{\gamma,\rho_\gamma}\geq \ps{H_2[\X-m_{\bm z,\bm a}(\X)]}_{\bm z,\bm a}+C$ for some $C$ independent of $\gamma$. Hence 
	\[
		\expect{e^{-R_{\gamma,\rho_\gamma}}}\leq c e^{c\gamma^{-\frac18}-c\gamma^{-\frac38}}+ C\expect{\exp\left(-\ps{H_{2}[\X-m_{\bm z,\bm a}(\X)]}_{\bm z,\bm a}\right)\mathds 1_{\mc A^c}}
	\]
	(see~\cite{LRV_semi1, LRV_semi2} for more details on the computations). Now the last expectation term above is finite thanks to Lemma~\ref{lemma:unif_DGMC}. This shows validity of the bound $\sup\limits_{\gamma\in(0,\frac12)}\expect{e^{-R_{\gamma,\rho_\gamma}}}<\infty$ so it remains to compare, for a given (small) $\gamma$, $R_{\gamma,\rho_\gamma}$ with $R_{\gamma,0}$. For this we consider
	\[
		\mc B\coloneqq\left\{\ps{\X-\X_{\rho_\gamma}}_{\bm z,\bm a}\geq \gamma^2\right\}\cup\left\{\ps{M_{\gamma,0}-M_{\gamma,\rho_\gamma}}_{\bm z,\bm a}\geq \gamma^3\right\}
	\]
	where, for $\rho\geq 0$, $M_{\gamma,\rho}\d v_g=e^{\gamma\X_{\rho}-\frac{\gamma^2}2\ln\rho}\d v_g$ on $\Sigma$ and $M_{\gamma,\rho}\d l_g=e^{\frac\gamma2\X_{\rho}-\frac{\gamma^2}4\ln\rho}\d l_g$ on $\partial\Sigma$. 
	Then by the exact same proof as the one of~\cite[Equation (3.28)]{LRV_semi1} (the arguments used there remain valid in our setting, see also~\cite[Lemma 5.3]{LRV_semi2}) we have the bound $\mathbb P(\mc B)\leq ce^{-c\gamma^{-2-\frac18}}$. Thus
	\begin{align*}
		\expect{e^{-R_{\gamma,0}}}&=\expect{e^{-R_{\gamma,\rho_\gamma}}e^{\left(R_{\gamma,\rho_\gamma}-R_{\gamma,0}\right)}\mathds1_{\mc B^c}}+\expect{e^{-R_{\gamma,0}}\mathds1_{\mc B}}
			\leq C\expect{e^{-R_{\gamma,\rho_\gamma}}}+ce^{c\gamma^{-2}-c\gamma^{-2-\frac18}}
	\end{align*}
	which is uniformly bounded near $\gamma=0$. The arguments are still true when we consider $\X+c$ instead of $\X$, thus concluding the proof.
    \end{proof}
	
	\subsubsection{Massive Gaussian Free Field}
    We now turn to the definition of the massive GFF that arises in the semi-classical limit of boundary Liouville theory. To this end we will need the following facts, which are classical (see \textit{e.g.}~\cite[Chapter 5]{Taylor_PDE}) when the underlying surface is smooth. And though the surface $(\Sigma,\gdiv)$ considered here has conical singularities and corners, they remain valid in this context too. The reader may consult \textit{e.g.}~\cite{MaWe} for more details.

	Let us consider the positive quadratic form
	\begin{equation}
		Q_{\Lambda,\sigma}(f,g)\coloneqq \frac1{2\pi}\int_\Sigma \left(\ps{\nabla_{\gdiv}f,\nabla_{\gdiv}h}_{\gdiv}^2+\Lambda fh\right)\d v_{\gdiv}+\frac1{2\pi}\int_{\partial\Sigma}\sigma fh\d l_{\gdiv}
	\end{equation}
	with domain $\dot{H}^1(\Sigma,\gdiv)\coloneqq\{f\in H^1(\Sigma,\gdiv),m_{\bm z,\bm a}(f)=0\}$. It defines a self-adjoint operator 
	\begin{equation}
		D_{\Lambda,\sigma}=\Lambda-\Delta_{\gdiv}\qt{with domain}\mc D_{\Lambda,\sigma}\coloneqq\left\{f\in \dot H^{2}(\Sigma,\gdiv)\text{ with }\left(\sigma+\partial_{n_{\gdiv}}\right)f=0\right\},
	\end{equation}
	\textit{i.e.} a massive Laplacian with Robin boundary conditions. 
    
	$D_{\Lambda,\sigma}$ admits a complete basis of orthonormal (for the $L^2(\Sigma,\gdiv)$ scalar product) eigenfunctions $(e_n)_{n\geq0}$ with (strictly) positive ordered eigenvalues $(\lambda_{\Lambda,\sigma,n})_{n\geq0}$. Indeed the arguments developed in~\cite[Chapter 5, Section 1 and Section A]{Taylor_PDE} apply verbatim to the case where the metric has conical singularities and where the boundary conditions are Robin ones, by considering the positive quadratic form $Q_{\Lambda,\sigma}$ instead of $(-\Delta u,u)$ there. This shows that $D_{\Lambda,\sigma}$ has an inverse $T$ which is positive and self-adjoint over $L^2(\Sigma,\gdiv)$. Moreover since from Proposition~\ref{prop:moser_trudinger} we have a compact embedding $\dot H^1(\Sigma,\gdiv)\hookrightarrow L^2(\Sigma,\gdiv)$ we know that $T$ is compact. As such there is an orthonormal basis of $L^2(\Sigma,\gdiv)$ that consists of eigenfunctions of $T$, and thus of $D_{\Lambda,\sigma}$. We can then define a Green's function $G_{\bm z,\bm a}(x,y)=\sum_{n\geq 1}\frac{e_n(x)e_n(y)}{\lambda_{\Lambda,\sigma,n}}$, which is such that for any $f$ smooth over $\bar\Sigma$ and any $x\in\Sigma$:
	\begin{eqs}
		&f(x)-m_{\bm z,\bm a}(f)=\\
		&\frac{1}{2\pi}\int_\Sigma G_{\bm z,\bm a}(x,y)\left(\Lambda-\Delta_{\gdiv}\right)f(y)\d v_{\gdiv}(y)+\frac{1}{2\pi}\int_\Sigma G_{\bm z,\bm a}(x,y)\left(\sigma+\partial_{n_{\gdiv}}\right)f(y)\d l_{\gdiv}(y).
	\end{eqs}
    Likewise we can define a positive quadratic form $\delta_{\Lambda,\sigma}$ with domain $\dot{H}^{-1}(\Sigma,\gdiv)$
    \begin{equation}
        \delta_{\Lambda,\sigma}(f,h)\coloneqq \ps{(\dot Gf)h}_{\bm z,\bm a},\qt{where}(\dot Gf)(x)\coloneqq\frac1{2\pi}\int_\Sigma \dot G(x,y)f(y)\d v_g(y).
    \end{equation}
    Here $\dot G$ is Neumann's Green function normalized to have $\ps{\dot G(x,\cdot)}_{\bm z,\bm a}=0$ for all $x\in\bar\Sigma$. Then 
    \begin{equation}
        \delta_{\Lambda,\sigma}(e_k,e_l)=\sum_{n\geq 1}Q_{\Lambda,\sigma}(e_k,e_n)\ps{e_n, (\dot G e_l)}_{\gdiv}-\delta_{k,l}
    \end{equation}
     for any positive integers $k,l$, and $\delta_{0,0}(e_k,e_l)=0$.
	\begin{lemma}\label{lemma:spec_mass}
		We have $\sum_{n\geq 1}\frac{1}{\lambda_{\Lambda,\sigma,n}^2}<\infty$. Likewise set $\delta\coloneqq\left(\delta_{\Lambda,\sigma}(e_k,e_l)\right)_{k,l\geq 1}$. Then 
        \begin{equation}
            0<\det\left((\rm I+\delta)e^{-\delta}\right)<\infty.
        \end{equation}
    \end{lemma}
	\begin{proof}
		The first point follows from $\int_{\Sigma^2}\norm{G_{\bm z,\bm a}(x,y)}^2\d v_{\gdiv}(x)\d v_{\gdiv}(y)<\infty$, which is true since  $G_{\bm z,\bm a}(x,\cdot)$ belongs to any $L^q(\Sigma,g)$ for $q<\infty$ while $e^{\phi_{\bm z,\bm a}}$ to $L^p(\Sigma,g)$ for some $p>1$. Alternatively this follows from Weyl's law for the Laplacian on surfaces with conical singularities proved \textit{e.g.} in~\cite{BS87, Che83}. The same argument applies to $\delta_{\Lambda,\sigma}$, showing that the corresponding self-adjoint operator is Hilbert-Schmidt, and thus $\rm{Tr}(\delta_{\Lambda,\sigma}^2)<\infty$. Choosing the orthonormal basis of $L^2(\Sigma,\gdiv)$ given by the $(e_k)_{k\geq 1}$ shows that the matrix $\delta$ satisfies $\rm{Tr}(\delta^2)<\infty$, which in turn implies that $\ln\det((I+\delta)e^{-\delta})$ is finite.
	\end{proof}
	Based on these estimates we can show the following:
	\begin{lemma}\label{lemma:unif_DGMC}
		The following quantities are finite and positive for any $c\in\R$: 
		\begin{equation}
			\expect{\exp\left(-\ps{H_{2}[\X+c]}_{\bm z,\bm a}\right)},\quad\expect{\exp\left(-\ps{H_{2}[\X-m_{\bm z,\bm a}(\X)]}_{\bm z,\bm a}\right)}.
		\end{equation}
        Moreover set $w_{\bm z,\bm a}\coloneqq\ps{W+\expect{\left(m_{\bm z,\bm a}(\X)\right)^2}}_{\bm z,\bm a}$ (recall $W$ from Lemma~\ref{lemma:approx_green}). Then
        \begin{equation}
            \expect{\exp\left(-\ps{H_{2}[\X-m_{\bm z,\bm a}(\X)]}_{\bm z,\bm a}\right)}=e^{-w_{\bm z,\bm a}}\det\left((\rm I+\delta)e^{-\delta}\right)^{-\frac12}.
        \end{equation}
	\end{lemma}
	\begin{proof}
		We write $\bar\X\coloneqq \X-m_{\bm z,\bm a}(\X)$. First of all, from Lemma~\ref{lemma:approx_green} and by definition of $H_2$ we have $\ps{H_2[\bar\X]}_{\bm z,\bm a}=\ps{\bar\X^2-\expect{\bar\X^2}}_{\bm z,\bm a}+w_{\bm z,\bm a}$.
        Let $(e_n)_{n\geq0}$ be the sequence of eigenfunctions of $D_{\Lambda,\sigma}$ described above and set $x_n\coloneqq \ps{\bar\X, e_n}_{\gdiv}$ for $n\geq 1$: they are centered Gaussian random variables with covariance kernel $\expect{x_kx_l}=Q_{0,0}(e_k,e_l)$. Moreover for any $N\geq 1$
		\[
		      \bar\X=\bar\X_{N}+\bar\X_{>N}\coloneqq\sum_{n=1}^Nx_ne_n+\sum_{n\geq N+1}x_ne_n
		\]
		where the second sum converges in $H^{-1}(\Sigma,g)$ almost surely via Lemma~\ref{lemma:spec_mass}. Now by Gaussian computations and using the definition of $\ps{\cdot}_{\bm z,\bm a}$ we can evaluate
		\[
		\expect{\exp\left(-\ps{\bar\X_N^2}_{\bm z,\bm a}\right)}=\frac{\det\left(D_{0,N}\right)^{\frac12}}{\det\left(D_{{\Lambda,\sigma},N}\right)^{\frac12}}\qt{and}\expect{\ps{\bar\X_N^2}_{\bm z,\bm a}}=\frac12\text{Tr}\left(D_{{\Lambda,\sigma},N}D_{0,N}^{-1}-\mathrm I_N\right)
		\]
		where $D_{0,N}$ and $D_{\Lambda,\sigma}$ are the $N\times N$ (Gram thus invertible) matrices with entries $\left(D_{0,N}\right)_{k,l}=Q_{0,0}(e_k,e_l)$ and $\left(D_{\Lambda,\sigma,N}\right)_{k,l}=Q_{\Lambda,\sigma}(e_k,e_l)=\lambda_k\delta_{k,l}$. Hence 
		\[
		\expect{\exp\left(-\ps{\bar\X_N^2-\expect{\bar\X_N^2}}_{\bm z,\bm a}\right)}=\det\left((\mathrm I_N+\delta_N)e^{-\delta_N}\right)^{-\frac12}\text{ with }\delta_N\coloneqq D_{\Lambda,\sigma,N}D_{0,N}^{-1}-\mathrm{I}_N.
		\]
        Thanks to Lemma~\ref{lemma:spec_mass}, such Gaussian computations thus allow to infer that
		\[
		      \expect{\exp\left(-\sum_{n,m\geq 1}\left(x_nx_m-\expect{x_nx_m}\right)\ps{e_ne_m}_{\bm z,\bm a}\right)}=\det\left((\mathrm I+\delta)e^{-\delta}\right)^{-\frac12}
		\]
        where the latter is finite and non-zero. To conclude it thus remains to show that indeed 
        \[
            \expect{\exp\left(-\sum_{n,m\geq 1}\left(x_nx_m-\expect{x_nx_m}\right)\ps{e_ne_m}_{\bm z,\bm a}\right)}=\expect{\exp\left(-\ps{\bar\X^2-\expect{\bar \X^2}}_{\bm z,\bm a}\right)}
        \]
        where the right-hand side is defined via $\lim\limits_{\rho\to0}\ps{\bar\X_\rho^2-\expect{\bar \X_\rho^2}}_{\bm z,\bm a}$. For this we use the almost sure limit $\lim\limits_{\rho\to0}\ps{\bar\X_\rho^2-\expect{\bar \X_\rho^2}}_{\bm z,\bm a}=\sum_{n,m\geq 1}\left(x_nx_m-\expect{x_nx_m}\right)\ps{e_ne_m}_{\bm z,\bm a}$ together with the uniform bound $\sup\limits_{\rho>0}\expect{\exp\left(-\ps{\bar\X_\rho^2-\expect{\bar \X_\rho^2}}_{\bm z,\bm a}\right)}<\infty$ which follows from finiteness of $\expect{\exp\left(-\ps{\bar\X^2-\expect{\bar \X^2}}_{\bm z,\bm a}\right)}$. Recollecting terms gives the second item.
	\end{proof}
	As we now expect from the previous arguments the latter gives rise to a massive GFF:
	\begin{lemma}\label{lemma:massive_GMC}
		Define a probability measure $\mathbb P_{\bm z,\bm a}$ over $H^{-1}(\Sigma,\gdiv)$ whose Radom-Nikodym derivative with respect to that of $\X$ is given by $\d\mathbb P_{\bm z,\bm a}\propto \exp\left(-\ps{H_{2}[\X-m_{\bm z,\bm a}(\X)]}_{\bm z,\bm a}\right)\d\mathbb P$. Then, under $\mathbb P_{\bm z,\bm a}$, $\X-m_{\bm z,\bm a}(\X)$ is a Gaussian field with covariance kernel given by $G_{\bm z,\bm a}$.
	\end{lemma}
	\begin{proof}
		Set $\ps{f,h}_{\Lambda,\sigma}\coloneqq \frac1{2\pi}\int_\Sigma f(\Lambda-\Delta_{\gdiv})h\d v_{\gdiv}+\frac1{2\pi}\int_{\partial\Sigma}f(\sigma+\partial_{n_{\gdiv}})h\d l_{\gdiv}$ and let $f$ be smooth over $\overline{\Sigma}$ with $m_{\bm z,\bm a}(f)=0$. Then using Girsanov's theorem we get
		\begin{align*}
			\expect{e^{\ps{\X,f}_{\Lambda,\sigma}}e^{-\ps{H_{2}[\X-m_{\bm z,\bm a}(\X)]}_{\bm z,\bm a}}}=\expect{e^{2\ps{(\X+f) f}_{\bm z,\bm a}}e^{-\ps{H_{2}[\X+f-m_{\bm z,\bm a}(\X+f)]}_{\bm z,\bm a}}}e^{\frac1{4\pi}\int_\Sigma\norm{\nabla f}^2\d v}
		\end{align*}
		Simplifying the latter we conclude that $\E_{\bm z,\bm a}\left[e^{\ps{\X ,f}_{\Lambda,\sigma}}\right]=e^{\frac12\ps{f,f}_{\Lambda,\sigma}}$ as expected  .
	\end{proof}
	We will call the field $\X_{\bm z,\bm a}$ thus defined a massive GFF in the metric $\gdiv$ with mass $\Lambda$ and Robin boundary conditions $(\sigma+\partial_{n_{\gdiv}})\X_{\bm z,\bm a}=0$.
	
	
	
	\subsection{The semi-classical limit}
	The classical theory pertaining Liouville CFT should be recovered by taking the coupling constant $\gamma$ to $0$, corresponding to the semi-classical limit. 
	We show that this is indeed the case and are actually more precise since we are able to write a two-terms expansion for the Liouville field $\Phi$. This expansion takes (formally) the form $\Phi=\frac1\gamma \phidiv +\tilde\Phi$ as $\gamma\to0$, where $\phidiv$ is the classical field and $\tilde\Phi$ a massive GFF: 
	\begin{theorem}\label{thm:semi_classical}
		Let $\Div$ be a divisor on $\Sigma$ with $a_k,b_l>-1$ for all $k,l$ and $\eul<0$. Let $\Lambda\not\equiv 0$ (resp.  $\sigma$) be non-negative and continuous over $\Sigma$ (resp. $\partial\Sigma\setminus\bm z$), and $\phidiv$ be the corresponding solution for Problem~\ref{prob:class}. Then for any $F$ continuous bounded over $H^{-1}(\Sigma,g)$
		\begin{eqs}\label{eq:semi_classical}
			&\ps{F\left(\Phi-\frac1\gamma\phidiv\right)\prod_{k=1}^{N}V_{-\frac{2a_k}\gamma}(z_k)\prod_{l=1}^{M}V_{-\frac{2\beta_l}\gamma}(s_l)}_{\gamma,\frac{\Lambda}{\gamma^2},\frac{\sigma}{\gamma^2}}=\frac{e^{-\frac1{\gamma^2}S_{\bm z,\bm a}}}{\mc Z_g}[F]_\gamma,\qt{where}\\
			&\lim\limits_{\gamma\to 0}[F]_\gamma=e^{-\frac12\chi(\Sigma)c_*} \int_\R\expect{F(\X+c)e^{-\ps{H_2(\X+c)}_{\bm z,\bm a}}}\d c.
		\end{eqs}
	\end{theorem}
	In the statement we have denoted $H_2(\X+c)=H_2(\X)+2cH_1(\X)+c^2$. The partition function $[1]_0$ can be evaluated thanks to Lemma~\ref{lemma:unif_DGMC}.
	An immediate consequence of the above statement together with Lemma~\ref{lemma:massive_GMC} is the following:
	\begin{corollary}
		The assignment for $F$ continuous bounded over $H^{-1}(\Sigma,g)$ of
		\begin{equation}
			\ps{F}_{\bm z,\bm a}\coloneqq \lim\limits_{\gamma\to0}\frac{\ps{F\left(\Phi-\frac1\gamma\phidiv\right)\V}_{\gamma,\frac{\Lambda}{\gamma^2},\frac{\sigma}{\gamma^2}}}{\ps{\V}_{\gamma,\frac{\Lambda}{\gamma^2},\frac{\sigma}{\gamma^2}}}
		\end{equation}
		describes the law of $\X_{\bm z,\bm a}+C$ where $\X_{\bm z,\bm a}$ is a massive GFF with covariance kernel $G_{\bm z,\bm a}$ (see Lemma~\ref{lemma:massive_GMC}) and $C$ is a centered Gaussian with variance $\frac1{\ps{2}_{\bm z,\bm a}}\cdot$
	\end{corollary}
	\begin{proof}
		By definition of the Liouville field, under the assumptions of Theorem~\ref{thm:semi_classical}
		\begin{align*}
			&\ps{F\left(\Phi-\frac1\gamma\phidiv\right)\prod_{k=1}^{N}V_{-\frac{2a_k}\gamma}(z_k)\prod_{l=1}^{M}V_{-\frac{2\beta_l}\gamma}(s_l)}_{\gamma,\frac{\Lambda}{\gamma^2},\frac{\sigma}{\gamma^2}}=e^{-\frac1{\gamma^2}G(\bm z,\bm a)}\int_\R e^{-\frac2\gamma \chi_\gamma c}\d c\\
			&\hspace{1cm}\expect{F\Big(\Hdiv+\gamma\left(\X+c\right)\Big)\exp\left(-\frac1{\gamma^2}\left(\frac1{2\pi}\int_\Sigma \Lambda e^{\gamma\left(\X+c\right)}\d v_{\gdiv}+\frac2{\pi}\int_{\partial\Sigma} \sigma e^{\frac\gamma2\left(\X+c\right)}\d l_{\gdiv}\right)\right)}
		\end{align*}
		where recall $G(\bm z,\bm a)$ from Proposition~\ref{prop:def_action} and $\chi_\gamma\coloneqq \eul+\frac{\gamma^2}{4}\chi(\Sigma)$. Now let $\phidiv$ be the solution of Problem~\ref{prob:class} given by Theorem~\ref{thm:prescribe} and write it under the form $\phidiv=\vphidiv+\Hdiv$, with $\vphidiv=\X_*+c_*$ for $c_*=m_g(\vphidiv)$. We can then combine the change of variable $c\to c+\frac1\gamma c_*$ in the integral together with Girsanov's theorem for
		\begin{align*}
			:e^Z:\hspace{0.2cm}=e^{Z-\frac1{2}\expect{Z^2}},\qt{with}Z\coloneqq -\frac 1\gamma\left(\frac{1}{2\pi}\int_\Sigma \X(-\Delta\X_*)\d v_g+\frac{1}{2\pi}\int_{\partial\Sigma} \X(\partial_{n_g}\X_*)\d l_g\right),
		\end{align*}
		inducing a shift of $\frac1\gamma\X_*$, to rewrite the correlation functions under the form
		\begin{align*}
			&\ps{F\left(\Phi-\frac1\gamma\phidiv\right)\prod_{k=1}^{N}V_{-\frac{2a_k}\gamma}(z_k)\prod_{l=1}^{M}V_{-\frac{2\beta_l}\gamma}(s_l)}_{\gamma,\frac{\Lambda}{\gamma^2},\frac{\sigma}{\gamma^2}}=e^{-\frac1{\gamma^2}\left(G(\bm z,\bm a)+2\chi_\gamma c_*\right)}\int_\R e^{-\frac2\gamma \chi_\gamma c}\d c\\
			&\E\Bigg[F\Big(\phidiv+\gamma\left(\X+c\right)\Big)
			e^{-\frac1{\gamma^2}\left(\frac1{2\pi}\int_\Sigma \Lambda e^{\gamma\left(\X+c\right)}e^{\vphidiv}\d v_{\gdiv}+\frac2{\pi}\int_{\partial\Sigma} \sigma e^{\frac\gamma2\left(\X+c\right)}e^{\frac12\vphidiv}\d l_{\gdiv}\right)}:e^{Z}:\Bigg].
		\end{align*}
		Now by definition of $\vphidiv$ we know that 
		\[
		Z=\frac{1}{2\pi}\int_\Sigma \left(-\frac 1\gamma\X\right)\Lambda e^{\vphidiv}\d v_{\gdiv}+\frac{2}{\pi}\int_{\partial\Sigma} \left(-\frac 1{2\gamma}\X\right)\sigma e^{\frac12\vphidiv}\d v_{\gdiv}.
		\]
		Moreover the variance of $Z$ is given by
		\[
		\expect{Z^2}=\frac1{\gamma^2}\frac1{2\pi}\int_\Sigma \norm{\nabla_g\X_*}^2_g\d v_g.
		\]
		As a consequence we arrive to the following equality:
		\begin{eqs}
			&\ps{F\left(\Phi-\frac1\gamma\phidiv\right)\prod_{k=1}^{N}V_{\alpha_k}(z_k)\prod_{l=1}^{M}V_{\beta_l}(s_l)}_{\gamma,\mu,\mu_\partial}=e^{-\frac1{\gamma^2}S_{\bm z,\bm a}}[F]_\gamma,\qt{where}\\
			&[F]_\gamma\coloneqq e^{-\frac12\chi(\Sigma) c_*} 
			\int_\R e^{-\frac2\gamma \chi_\gamma c}\d c\\
			&\hspace{2.5cm}\E\Bigg[F\left(\X+c\right)e^{-\frac1{\gamma^2}\left(\frac1{2\pi}\int_\Sigma  \left(e^{\gamma\left(\X+c\right)}-1-\gamma\X\right)\d v_*+\frac2\pi\int_{\partial\Sigma}\left( e^{\frac\gamma2\left(\X+c\right)} -1-\frac\gamma2\X\right)\d l_*\right)}\Bigg]
		\end{eqs}
		and with $S_{\bm z,\bm a}=G(\bm z,\bm a)+I_{\bm z,\bm a}$ being exactly the Liouville action as given by Proposition~\ref{prop:def_action}. 
		
		Therefore it remains to investigate the limit as $\gamma\to0$ of $[F]_\gamma$. To do so we first note that thanks to the Gauss-Bonnet formula from Proposition~\ref{prop:gauss-bonnet} 
		\[
		\frac1{2\pi}\int_\Sigma \Lambda e^{\vphidiv}\d v_{\gdiv}+\frac2\pi\int_{\partial\Sigma}\sigma e^{\frac12\vphidiv}\d l_{\gdiv}=2\eul,\qt{so that}
		\]
		\begin{align*}
			&[F]_\gamma= e^{-\frac12\chi(\Sigma)c_*}
			\int_\R e^{-\frac\gamma2\chi(\Sigma)c}\E\Bigg[F\left(\X+c\right)e^{-\frac1{4\pi}\int_\Sigma  \tilde{\mc D}_\gamma[\X+c]\d v_*-\frac1{4\pi}\int_{\partial\Sigma}\tilde{\mc D}^\partial_\gamma[\X+c]\d l_*}\Bigg]\d c.
		\end{align*}
		Now thanks to Lemma~\ref{lemma:lim_der_GMC} we know that the expectation term in the above converges to the one in Equation~\eqref{eq:semi_classical}. To ensure the uniform convergence of the integral in $c$ we can make the change of variable $c\to c-m_{\bm z,\bm a}(\X)$ to rewrite it as, up to a $o(1)$ term,
		\[
		\int_\R\expect{F(\overline \X+c)e^{-\frac{1}{4\pi}\int_\Sigma H_2(\overline\X)\d v_*-\frac{1}{4\pi}\int_{\partial\Sigma} H_2(\overline\X)\d l_*}}e^{-\frac{c^2}{2V}}\d c\qt{with}
		\]
		$\overline{\X}\coloneqq \X-m_{\bm z,\bm a}(\X)$ and $\frac1{2V}=\ps{1}_{\bm z,\bm a}$. The proof is thus concluded thanks to Lemma~\ref{lemma:unif_DGMC}.
	\end{proof}

	
	
	
	\section{Implications of the semi-classical limit}\label{sec:accessory}
	In the previous section we have shown that for $\gamma\in(0,2)$ Liouville CFT describes a random geometry that will concentrate on the classical one described in Section~\ref{sec:unif} in the semi-classical limit $\gamma\to0$. We provide here some of its classical implications in the case where $\Sigma=\H$:
	\begin{theorem}\label{thm:accessory_HEM}
		In the setting of Theorem~\ref{thm:SET} the accessory parameters that appear in the expansion of the stress-energy tensor $T$ are given, for $1\leq k\leq N$ and $1\leq l\leq M$, by
		\begin{equation}
			\bm c_k=-\frac12\partial_{z_k} S_{\bm z,\bm a}\qt{and}\bm c_l=-\frac12\partial_{s_l} S_{\bm z,\bm a}.
		\end{equation}
		They are subject to classical global Ward identities for $0\leq n\leq 2$:
		\begin{equation}
			\sum_{k=1}^{N}2\mathfrak{Re}\left(z_k^n\bm c_k+nz_k^{n-1}\delta_k\right)+\sum_{l=1}^{M}\left(s_l^n\bm c_l+ns_l^{n-1}\delta_l\right)=0.
		\end{equation}
		Moreover the classical field $\Phidiv$ is such that, in the weak sense of derivatives,
		\begin{equation}\label{eq:HEM_class}
			\left\lbrace \begin{array}{ll}
				\left(\partial_z^2+\frac12 T(z)\right)e^{-\frac12\Phidiv(z)}=0&\text{in }\H\setminus\bm z\\
				\left(\partial_t^2+2 T(t)\right)e^{-\frac14\Phidiv(t)}=\left(\sigma(t)^2-\frac{\Lambda}{2}\right)e^{\frac34\Phidiv(t)}&\text{on }\R\setminus\bm z.
			\end{array} \right.
		\end{equation}
	\end{theorem}
	Like in Section~\ref{sec:SET} the definition of the derivative of the Liouville action requires some care and has to be understood in the weak sense. We first explain how to address this issue and then show that the higher equations of motion~\eqref{eq:HEM_class} are valid. To do so we will strongly rely on the semi-classical statement proved in the previous section. Recall that to do so we have considered the scalings $\mu=\frac{\Lambda}{\gamma^2}$ and $\mu_\partial=\frac{\sigma}{\gamma^2}$ as well as $\alpha_k=-\frac{2a_k}{\gamma}$ as $\gamma\to0$. 
	
	
	
	\subsection{Accessory parameters}
	In order to compute the accessory parameters $\bm c_k$ we first need to gather some information about Liouville CFT and more specifically on the descendant at order $1$ $\mc L_{-1}$ from~\cite{Cer_HEM}, which is the quantum analog of the $\Lc_{-1}$ considered before.
	
	\subsubsection{Descendants in Liouville CFT} For any positive $\delta,\eps,\gamma$, the regularized $\mc L_{-1}^{(l)}$ descendant is defined in~\cite[Equation (3.2)]{Cer_HEM} within correlation functions by setting
	\begin{eqs}\label{eq:L1_CFT}
		&\mc L_{-1}^{(l)}\ps{\prod_{k=1}^NV_{\alpha_k}(z_k)\prod_{m=1}^MV_{\beta_m}(s_m)}_{\gamma,\delta,\eps}=\sum_{k\neq l}\frac{\alpha_k\beta_l}{2(z_k-s_l)}\ps{\V}_{\delta,\eps}-\It\left[\frac{\gamma\beta_l}{2(t-s_l)}\right]
	\end{eqs}
	with $\ps{\V}_{\gamma,\delta,\eps}\coloneqq\ps{\prod_{k=1}^NV_{\alpha_k}(z_k)\prod_{m=1}^MV_{\beta_m}(s_m)}_{\gamma,\delta,\eps}$. By construction (see~\cite[Corollary 3.3]{Cer_HEM}) it satisfies for any smooth, compactly supported test function $f$ on $\R\setminus\left(\bm z\setminus\{s_l\}\right)$ 
	\begin{equation}\label{eq:IPP}
		\int_\R \mc L_{-1}^{(l)}\ps{\V}_{\gamma,\delta,\eps}f(s_l)\d s_l=-\int_\R \ps{\V}_{\gamma,\delta,\eps}\partial f(s_l)\d s_l.
	\end{equation}
	The $\mc L_{-1}^{(l)}$ descendant is defined within correlation functions by taking an appropriate limit $\delta,\eps\to0$ of the regularized one, see~\cite[Lemma 3.2]{Cer_HEM}.
	
	\subsubsection{Computation of the accessory parameters}
	Based on this knowledge we are now in position to compute the accessory parameters for the stress-energy tensor:
	\begin{proposition}\label{prop:accessory}
		For any $1\leq l\leq M+1$ the following limit exists:
		\begin{equation}
			\partial_{s_l}S_{\bm z,\bm a}\coloneqq\lim\limits_{\delta,\eps\to0}\left(\partial_{s_l}S_{\delta,\eps}-\eval{s_l+\eps}{s_l-\eps}{-}{2\sigma(t) e^{\frac12\phireg(t)}}\right).
		\end{equation}
		Moreover it coincides with the weak derivative of $S_{\bm z,\bm a}$ and we have the equality
		\begin{equation}
			\Lc_{-1}^{(l)}[\phidiv]=-\partial_{s_l}S_{\bm z,\bm a}.
		\end{equation}
		The same holds for a bulk insertion: $\Lc_{-1}^{(k)}[\phidiv]=-\partial_{z_k}S_{\bm z,\bm a}$ for $0\leq k\leq 2N+1$.
	\end{proposition}
	\begin{proof}
		For any $1\leq l\leq M+1$, let us take $f$ compactly supported on $\R\setminus\left(\bm z\setminus\{s_l\}\right)$.
		Then for any positive $\gamma,\delta,\eps$ in agreement with Equation~\eqref{eq:IPP}
		\begin{align*}
			\int_\R \frac{\mc L_{-1}^{(l)}\ps{\V}_{\gamma,\delta,\eps}}{\ps{\V}_{\gamma,\delta,\eps}}f(s_l)\d s_l=-\int_\R \ln\left(\ps{\V}_{\gamma,\delta,\eps}\right) \partial_{s_l}f(s_l)\d s_l
		\end{align*}
		Now as $\gamma\to0$, Theorem~\ref{thm:semi_classical} together with the explicit expressions for $\mc L_{-1}^{(l)}\ps{\V}_{\gamma,\delta,\eps}$ from Equation~\eqref{eq:L1_CFT} and for $\Lc_{-1}^{(l)}[\Phireg]$ (Equation~\eqref{eq:L1}) give (using the scalings in $\gamma$ of $\mu,\mu_{\partial},\alpha_k$)
		\[
		\frac{\mc L_{-1}^{(l)}\ps{\V}_{\gamma,\delta,\eps}}{\ps{\V}_{\gamma,\delta,\eps}}\sim \frac1{\gamma^2}\Lc_{-1}^{(l)}[\Phireg]\qt{and}\ln\left(\ps{\V}_{\gamma,\delta,\eps}\right)\sim-\frac1{\gamma^2}S_{\delta,\eps}.
		\]
		Since $f$ is compactly supported such asymptotics can be integrated against $f$, hence
		\begin{align*}
			\int_\R\Lc_{-1}^{(l)}[\Phireg]f(s_l)\d s_l=\int_\R S_{\delta,\eps} \partial_{s_l}f(s_l)\d s_l.
		\end{align*}
		This shows that for positive $\delta,\eps$, $\Lc_{-1}^{(l)}[\Phireg]=-\partial_{s_l}S_{\delta,\eps}$ in the weak sense of derivatives. The proof is thus concluded via Lemma~\ref{lemma:L1} by taking $\delta,\eps\to0$ in the above since 
		\[
		\int_\R \eval{s_l+\eps}{s_l-\eps}{-}{2\sigma(t) e^{\frac12\phireg(t)}}f(s_l)\d s_l= -\int_\R 2\sigma(t) e^{\frac12\phireg(t)}\eval{s_l+\eps}{s_l-\eps}{-}{f(t)}\hspace{0.2cm}\d s_l\to0.
		\]
	\end{proof}
	
	\subsubsection{Global Ward identities}
	The accessory parameters are subject to constraints, classical counterparts of the global Ward identities for the correlation functions~\cite[Theorem 3.7]{Cer_HEM}:
	\begin{equation}\label{eq:global_ward_CFT}
		\sum_{k=1}^{2N+M}\left(z_k^n\bm {\mathcal L_{-1}^{(k)}}+nz_k^{n-1}\Delta_{\alpha_k}\right)\ps{\prod_{k=1}^NV_{\alpha_k}(z_k)\prod_{l=1}^MV_{\beta_l}(s_l)}=0
	\end{equation}
	with the conformal weights $\Delta_{\alpha_k}=\frac{\alpha_k}{2}(Q-\frac{\alpha_k}{2})$. The semi-classical limit $\gamma\to0$ thus gives
	\[
	\lim\limits_{\gamma\to0}\frac{\gamma^2}{\V} \sum_{k=1}^{2N+M}\left(z_k^n\bm {\mathcal L_{-1}^{(k)}}+nz_k^{n-1}\Delta_{\alpha_k}\right)\ps{\V}=2\sum_{k=1}^{2N+M}\left(z_k^n\bm c_k+nz_k^{n-1}\delta_k\right).
	\]
	As a consequence the global Ward identities become in the semi-classical limit:
	\begin{equation}\label{eq:global_ward_class}
		\sum_{k=1}^{N}2\mathfrak{Re}\left(z_k^n\bm c_k+nz_k^{n-1}\delta_k\right)+\sum_{l=1}^{M}\left(s_l^n\bm c_l+ns_l^{n-1}\delta_l\right)=0.
	\end{equation}
	
	
	\subsection{Higher equations of motion}
	In the boundary case, the BPZ-type differential equations usually satisfied by the correlation functions of Liouville CFT may no longer be valid and may be given instead by \textit{higher equations of motion}~\cite{Za04, BB10, BaWu, Cer_HEM}. We obtain here their classical counterpart thanks to the semi-classical limit of boundary Liouville CFT.
	
	\subsubsection{Higher equations of motion in boundary Liouville CFT}
	Let us assume that one of the Vertex Operators has weight given by $\alpha=-\frac\gamma2$ (a similar result holds for $\alpha=-\frac2\gamma$ and gives rise to the classical HEM~\eqref{eq:HEM_heavy}). Then under generic assumptions on the cosmological constants the correlation functions satisfy the following equations~\cite[Theorem 1.3]{Cer_HEM}:
	\begin{eqs}\label{eq:HEM_CFT}
		&\left(\frac{4}{\gamma^2}\partial^2_z+\mc T(z)\right)\ps{V_{-\frac\gamma2}(z)\prod_{k=1}^NV_{\alpha_k}(z_k)\prod_{l=1}^M V_{\beta_l}(s_l) }=0\qt{for}z\in\H,\\
		&\left(\frac{4}{\gamma^2}\partial^2_t+\mc T(t)\right)\ps{V_{-\frac\gamma2}(t)\prod_{k=1}^NV_{\alpha_k}(z_k)\prod_{l=1}^M V_{\beta_l}(s_l) }=c_\gamma(\bm\mu)\ps{V_{\frac{3\gamma}{2}}(t)\V}\qt{for}t\in\R
	\end{eqs}
	where $\mc T$ is the (weak) differential operator defined by
	\begin{equation}
		\mc T(z)\coloneqq\sum_{k=1}^{2N+M}\frac{\Delta_{\alpha_k}}{(z-z_k)^2}+\frac{\partial_{z_k}}{z-z_k}\qt{for}z\in\overline{\H}
	\end{equation}
	with the conformal weights $\Delta_{\alpha_k}=\frac{\alpha_k}{2}(Q-\frac{\alpha_k}{2})$ and the constant $c_\gamma(\bm\mu)$ given by 
	\begin{equation}
		c_\gamma(\bm\mu)=\left(\mu_L^2+\mu_R^2-2\mu_L\mu_R\cos\left(\frac{\pi\gamma^2}{4}\right)-\frac\mu{2\pi}\sin\left(\frac{\pi\gamma^2}{4}\right)\right)\frac{\Gamma\left(\frac{\gamma^2}{4}\right)\Gamma\left(1-\frac{\gamma^2}{2}\right)}{\Gamma\left(1-\frac{\gamma^2}{4}\right)}
	\end{equation}
	for $\gamma<\sqrt2$ and by $0$ if $\gamma>\sqrt2$. Here $\mu_L=\frac{2\mu_\partial(t^-)}\pi$ and $\mu_R=\frac{2\mu_\partial(t^+)}\pi$.
	
	\subsubsection{Classical higher equations of motion} We now look at the semi-classical limit of the above, where recall the scalings $\mu=\frac{\Lambda}{\gamma^2}$ and $\mu_\partial=\frac{\sigma}{\gamma^2}$ as well as $\alpha_k=-\frac{2a_k}{\gamma}$. We assume that $\mu_L=\mu_R=\frac{\sigma(t)}{\gamma^2}$, so that the curvatures remain unchanged by the insertion of the Vertex Operator $V_{-\frac{\gamma}{2}}$. Now thanks to Theorem~\ref{thm:semi_classical} we have that
	\[
	\frac{c_\gamma(\bm\mu)\ps{V_{\frac{3\gamma}{2}}(t)\V}}{\ps{\V}}\sim \frac{4}{\gamma^2} \left(\left(\frac{\sigma(t)}{2}\right)^2-\frac{\Lambda}{8}\right)e^{\frac{3}{4}\Phidiv(t)}.
	\]
	Likewise thanks to Proposition~\ref{prop:accessory} together with Theorem~\ref{thm:semi_classical}
	\[
	\frac{\mc T(z)\ps{V_{-\frac\gamma2}(z)\V}}{\ps{\V}}\sim \frac1{\gamma^2}\sum_{k=1}^{2N+M}\left(\frac{-a_k(2+a_k)}{(z-z_k)^2}-\frac{\partial_{z_k} S_{\bm z,\bm a}}{z-z_k}\right)e^{-\frac14\Phidiv(t)}.
	\]
	As a consequence by taking the $\gamma\to0$ limit of $\frac{\gamma^2}{4\ps{\V}}$ times Equation~\eqref{eq:HEM_CFT} we get
	\begin{eqs}
		&\left(\partial_z^2+\frac12T(z)\right)e^{-\frac12\Phidiv(z)}=0\qt{for}z\in\H,\\
		&\left(\partial_t^2+\frac12T(t)\right)e^{-\frac14\Phidiv(t)}=\frac14\left(\sigma(t)^2-\frac{\Lambda}{2}\right)e^{\frac{3}{4}\Phidiv(t)}\qt{for}t\in\R.
	\end{eqs}
	This concludes for the proof of Theorem~\ref{thm:accessory_HEM}.

	\bibliography{biblio}
	\bibliographystyle{plain}
	
\end{document}